\pdfoutput=1

\documentclass[sigconf]{aamas} 

\usepackage{balance} 
\usepackage{notation}

\setcopyright{ifaamas}
\acmConference[AAMAS '21]{Proc.\@ of the 20th International Conference on Autonomous Agents and Multiagent Systems (AAMAS 2021)}{May 3--7, 2021}{Online}{U.~Endriss, A.~Now\'{e}, F.~Dignum, A.~Lomuscio (eds.)}
\copyrightyear{2021}
\acmYear{2021}  
\acmDOI{}
\acmPrice{}
\acmISBN{}

\acmSubmissionID{632}

\title{Equilibrium Refinements for Multi-Agent Influence Diagrams: Theory and Practice}

\author{Lewis Hammond}
\affiliation{
  \institution{University of Oxford}
  }
\email{lewis.hammond@cs.ox.ac.uk}

\author{James Fox}
\affiliation{
  \institution{University of Oxford}
  }
\email{james.fox@cs.ox.ac.uk}

\author{Tom Everitt}
\affiliation{
  \institution{DeepMind}
  }
\email{tomeveritt@google.com}

\author{Alessandro Abate}
\affiliation{
  \institution{University of Oxford}
  }
\email{aabate@cs.ox.ac.uk}

\author{Michael Wooldridge}
\affiliation{
  \institution{University of Oxford}
  }
\email{mjw@cs.ox.ac.uk}

\begin{abstract}
Multi-agent influence diagrams (MAIDs) are a popular form of graphical model that, for certain classes of games, have been shown to offer key complexity and explainability advantages over traditional extensive form game (EFG) representations. In this paper, we extend previous work on MAIDs by introducing the concept of a MAID subgame, as well as subgame perfect and trembling hand perfect equilibrium refinements. We then prove several equivalence results between MAIDs and EFGs. Finally, we describe an open source implementation for reasoning about MAIDs and computing their equilibria.

\end{abstract}


\newcommand{\BibTeX}{\rm B\kern-.05em{\sc i\kern-.025em b}\kern-.08em\TeX}

\begin{document}


\pagestyle{fancy}
\fancyhead{}


\maketitle 


\section{Introduction}
\label{sec:intro}
Multi-agent influence diagrams (MAIDs) are a compact and expressive graphical representation for non-cooperative games. Introduced by Koller and Milch (henceforth K\&M) \cite{koller2003multi, milch2008ignorable}, they offer three key advantages over the classic extensive form game (EFG) representation. First, MAIDs can depict many games more compactly than EFGs, especially those with incomplete information. 
Second, MAIDs encode conditional independencies between variables. This means large MAIDs can often be decomposed into several smaller ones, with potentially exponential speedups for finding Nash equilibria \cite{koller2003multi}.
Third, MAIDs often make it possible to explicitly represent aspects of game structure that are obscured in EFGs. While it is possible to convert any EFG to a MAID of at most the same size (Section \ref{sec:EFGtoMAID}), it is true that EFGs are sometimes better suited for modelling asymmetric decision problems. With that said, every model has its weaknesses, and how useful a particular representation is rests on its strengths. 
We further develop both the theory and practical tools for MAIDs in order to allow both researchers and practitioners to make the most of their strengths.

Previous work on MAIDs has focussed on Nash equilibria as the core solution concept \cite{nash1950equilibrium}. Whilst this is arguably the most important solution concept in non-cooperative game theory, if there are many Nash equilibria we often wish to remove some of those that are less `rational'. Many refinements to the Nash equilibrium have been proposed \cite{Maschler2013}, with two of the most important being subgame perfect Nash equilibria \cite{selten1965spieltheoretische} and trembling hand perfect equilibria \cite{selten1974reexamination}. The first rules out `non-credible' threats and the second requires that each player is still playing a best-response when other players make small mistakes. On the practical side, while much software exists for normal or extensive form games, there is no such implementation for reasoning about games expressed as MAIDs, despite their computational advantages.

\subsection{Contribution}

In this paper, we make the following contributions. First, we extend the applicability of MAIDs by introducing the concept of a MAID subgame (Section \ref{sec:subgames}) and build on this concept to introduce subgame perfect and trembling hand perfect equilibrium refinements (Section \ref{sec:eqrefine}). Second, we prove several equivalence results between MAIDs and EFGs, demonstrating the preservation of the key game-theoretic concepts described above when representing EFGs as MAIDs and thus further justifying the use of this model. These proofs are constructive and are based on procedures for converting between EFGs and MAIDs, the full details of which are included in Appendices \ref{sec:macid2efg} and \ref{sec:EFG2MAID}. Third, we report on our open source codebase for computing our equilibrium refinements in MAIDs (Section \ref{sec:implementation}).

\subsection{Related Work}

Our work builds primarily on the seminal work of K\&M \cite{koller2003multi, milch2008ignorable}. 
More recently, casual influence diagrams (CIDs) have been defined \cite{Everitt2021}, where the probabilistic arrows in influence diagrams are interpreted as describing a causal relationship, in accordance with Pearl's graphical causal models \cite{pearl2009causality}. CIDs model single agents, helping to predict behaviour by identifying the incentives that arise due to the agent optimising its objective, and have been shown to have many applications \cite{carey2020incentives,everitt2019modeling, everitt2019reward, holtman2020agi,Langlois2021}. Our equilibrium refinements for MAIDs and our implementation are partially targeted at extending this work on incentives to the multi-agent setting.

Pfeffer and Gal investigated when an agent is motivated to care about its decision in the context of MAIDs, identifying four reasoning patterns (with associated graphical criteria) that justify a particular decision choice \cite{pfeffer2007reasoning}. Later work showed practical applications of these reasoning patterns, which can lead to safer human-machine or machine-machine designs and again reduce the time complexity of computing Nash equilibria \cite{antos2012identifying}. In this work we implement these reasoning patterns in our codebase. Building on this, further research could consider which reasoning patterns arise when agents are playing a certain equilibrium refinement.

Several other formalisms, often partly inspired by MAIDs, have been proposed for representing and reasoning about games as probabilistic graphical models. For example: networks of influence diagrams represent mental models of the different agents as nodes in a graph and use these to describe and reason about belief structures \cite{Gal2008}; settable systems extend structural equation models to include the concept of optimisation and hence the idea of a `best response', which is key to defining game-theoretic equilibria \cite{White2009}; temporal action graph games are similar to MAIDs, but can be more compact for games that involve anonymity or context-specific utility independencies \cite{Jiang2009}. These works, however, focus on the introduction of novel representations, whereas we focus on deepening the theory and practice behind an existing representation. It is an interesting question for further research whether our insights also apply to these related models.



\section{Background}

In this section, we define EFGs and MAIDs and show how their graphical representation of games differ with the help of the following example \cite{spence1978job}.

\begin{example}[Job hiring]
\label{ex:hiring}

\textit{A company employs an AI system to automate their hiring process. A naturally hard-working or naturally lazy worker wants a job at this company and believes that a university degree will increase their chance of being hired; however, they also know that they will suffer an opportunity cost from three years of studying. A hard-worker will cope better with a university workload than a lazy worker. The algorithm must decide, on behalf of the company, whether to hire the worker. The company wants to hire someone who is naturally hard-working, but the algorithm can't observe the worker's temperament directly, it can only infer it indirectly through whether or not the worker attended university.}
\end{example}
 
We use capital letters $X$ for variables and let $\dom(X)$ denote the domain of X. An assignment $x \in \dom(X)$ to $X$ is an instantiation of $X$ denoted by $X=x$. $\bm{X} = \{X_1, \dots, X_n\}$ is a set of variables with domain $\dom(\bm{X}) = \times^{n}_{i=1}\dom(X_i)$ and $\bm{x} = \{x_1, \dots, x_n\}$ is the set containing an instantiation of all variables in $\bm{X}$. We let $\Pa_V$ denote the parents of a node $V$ in a graphical representation and $\pa_V$ be the instantiation of $\Pa_V$. $\Ch_V$, $\Anc_V$, $\Desc_V$, and $\Fa_V \coloneqq \Pa_V \cup \{V\}$ are the children, ancestors, descendants, and family of $V$ with, analogously to $\pa_V$, their instantiations written in lowercase. Unless otherwise indicated we index mathematical objects with superscripts $i \in \bm{N}$ to denote their affiliation with a player $i$ (where $\bm{N}$ is a set of players) and with subscripts $j \in \mathbb{N}$ to enumerate them. 

\subsection{Extensive Form Games}
\begin{definition}[\cite{Kuhn1953}]
    An \textbf{extensive-form game (EFG)} $\efg$ is a tuple \\ $(\bm{N}, T, \bm{P}, \bm{D}, \lambda, \bm{I}, U),  \textrm{where:}$
    \begin{itemize}
        \item $\bm{N} = \{1,\dots,n\}$ is a set of agents.

        \item $T = (\bm{V}, \bm{E})$ is a game tree with nodes $\bm{V}$ 
        that are partitioned into the sets $\bm{V}^0, \bm{V}^1, \dots, \bm{V}^n, \bm{L}$ where $R \in \bm{V}$ is the root of $T$, $\bm{L}$ is the set of leaves of the tree, $\bm{V}^0$ is the set of chance nodes, and $\bm{V}^i$ is the set of nodes controlled by player $i \in \bm{N}$. These nodes are connected by edges $\bm{E} \subseteq \bm{V} \times \bm{V}$.
        
        \item $\bm{P} = \{P_1,\dots,P_{\vert\bm{V}^0\vert}\}$ is a set of probability distributions where each $P_j : \Ch_{V_j} \rightarrow [0,1]$ determines the probability of a path through the game tree that has reached chance node $V^0_j$ proceeding to each child node in $\Ch_{V^0_j}$.
        
        \item $\bm{D}$ is a set of decisions, we write $\bm{D}^i_j \subseteq \bm{D}$ to describe the set of available decisions at node $V^i_j \in \bm{V}^i$.
       
        \item $\lambda : \bm{E} \rightarrow \bm{D}$ is a labelling function mapping an edge $(V^i_j, V^k_l)$ to a decision $d \in \bm{D}^i_j$.
    
        \item $\bm{I} = \{\bm{I}^1,\dots,\bm{I}^n\}$ is a set such that for each player $\bm{I}^i \subset 2^{\bm{V}^i}$ defines a partition of the vertices controlled by player $i$ into information sets.
        
        \item $U : \bm{L} \rightarrow \mathbb{R}^n$ is a utility function mapping each leaf node to a vector that determines the final payoff for each player.
    \end{itemize}

An \textbf{information set} $I_j^i \in \bm{I}^i$ is defined such that for all $V^i_k, V^i_l \in I_j^i$ we have $\bm{D}^i_j \coloneqq \bm{D}^i_k = \bm{D}^i_l$. In other words, the same player $i$ selects the decision and the same decisions are available at each of the nodes in an information set. When $\vert I_j^i\vert = 1$ for all $i$ and $j$, $\efg$ is a \textbf{perfect information} game. 
A (behavioural) \textbf{strategy} $\sigma^i$ for a player $i$ is a set of probability distributions $\sigma_j^i : \bm{D}^i_j \rightarrow [0,1]$ over the actions available to the player at each of their information sets $I_j^i$.\footnote{Formally, $\sigma_j^i(d) > 0$ only if $d \in D^i_k$ for any vertex $V_k \in I_j^i$, and $\Sigma_{d \in D^i_k}\sigma_j^i(d) = 1$.} A strategy is \textbf{pure} when $\sigma_j^i(d) \in \{0,1\}$ for all information sets $I_j^i$ and \textbf{fully mixed} when $\sigma_j^i(d) > 0$ for all $d \in \bm{D}^i_j$.  A \textbf{strategy profile} $\sigma = (\sigma^1, \dots, \sigma^n)$ is a tuple of strategies one for each player $i \in \bm{N}$. $\sigma^{-i} = (\sigma^1, \dots, \sigma^{i-1}, \sigma^{i+1},\dots, \sigma^n)$ denotes the partial strategy profile of all players other than $i$, and so $\sigma = (\sigma^i, \sigma^{-i})$. The combination of the distributions in $\bm{P}$ with a strategy profile $\sigma$ thus defines a full probability distribution $P^\sigma$ over paths in $\mathcal{G}$.
\end{definition}

\begin{figure}[h]
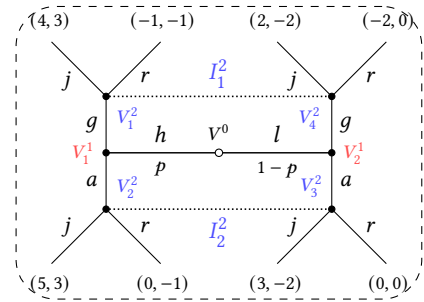

    {\centering
    \begin{istgame}[scale=0.75]
        \xtdistance{20mm}{20mm}
        \istroot(0)[chance node]{$V^0$}
        \istb<grow=left>{h}[a]
        \istb<grow=right>{l}[a]
        \istb<grow=left>{\text{\footnotesize $p$}}[b]
        \istb<grow=right>{\text{\footnotesize $1-p$}}[b]
        \endist
        \xtdistance{10mm}{20mm}
        \istroot(1)(0-1)<180, red!70>{$V^1_1$}
        \istb<grow=north>{g}[l]
        \istb<grow=south>{a}[l]
        \endist
        \istroot(2)(0-2)<0, red!70>{$V^1_2$}
        \istb<grow=north>{g}[r]
        \istb<grow=south>{a}[r]
        \endist
        \istroot'[north](a1)(1-1)<315, blue!70>{$V^2_1$}
        \istb{j}[bl]{(4,3)}
        \istb{r}[br]{(-1,-1)}
        \endist
        \istroot(b1)(1-2)<45, blue!70>{$V^2_2$}
        \istb{j}[al]{(5,3)}
        \istb{r}[ar]{(0,-1)}
        \endist
        \istroot(a2)(2-2)<100, blue!70>{$V^2_3$}
        \istb{j}[al]{(3,-2)}
        \istb{r}[ar]{(0,0)}
        \endist
        \istroot'[north](b2)(2-1)<225, blue!70>{$V^2_4$}
        \istb{j}[bl]{(2,-2)}
        \istb{r}[br]{(-2,0)}
        \endist
        \xtInfoset(a1)(b2){\textcolor{blue!70}{$I^2_1$}}
        \xtInfoset(b1)(a2){\textcolor{blue!70}{$I^2_2$}}[below]
        \xtSubgameBox(1){(1-1)(1-2)(2-1)(2-2)(a1-1)(a1-2)(b1-1)(b1-2)(a2-1)(a2-2)(b2-1)(b2-2)}[black,inner sep = 12pt]
        \end{istgame}}
    \caption{An EFG representation of Example \ref{ex:hiring}.}
    \label{fig:signalEFG}
\end{figure}

Figure \ref{fig:signalEFG} shows Example \ref{ex:hiring}'s signalling game in extensive form. Nature, as a chance node $V^0$, flips a biased coin at the root of the tree to decide whether the person is hard-working (probability $p$) or lazy (probability $1-p$). The worker (player 1)'s decision whether to go ($g$) or avoid ($a$) university is represented at nodes $\{V^1_1, V^1_2\} = \bm{V}^1$ and there are four nodes $\{V^2_1, V^2_2, V^2_3, V^2_4\} = \bm{V}^2$ for the hiring algorithm (player 2), each with two decision options: reject ($r$) or job offer ($j$). These nodes are split into two information sets (dotted lines between nodes) because the algorithm does not know whether the person is naturally hard-working. The payoffs for the worker and the employer respectively are given at the leaves of the tree.

\subsection{Multi-Agent Influence Diagrams}
\label{sec:MACID}
Following recent work \cite{howard2005influence, Everitt2021}, we depart slightly from the convention of K\&M to distinguish between an influence \emph{diagram}, which gives the structure of a strategic interaction, and an influence \emph{model}, which adds a particular parametrisation to the diagram.

\begin{definition}[\citep{koller2003multi}]
\label{def:MAID}
    A \textbf{multi-agent influence diagram (MAID)} is a triple $(\bm{N}, \bm{V}, \bm{E})$, where:
    \begin{itemize}
        \item $\bm{N} = \{1,\dots,n\}$ is a set of agents.
        \item $(\bm{V}, \bm{E})$ is a directed acyclic graph (DAG) with a set of vertices $\bm{V}$ connected by directed edges $\bm{E} \subseteq \bm{V} \times \bm{V}$. These vertices are partitioned into $\bm{D}$, $\bm{U}$, and $\bm{X}$, which correspond to decision, utility, and chance nodes respectively. $\bm{D}$ and $\bm{U}$ are in turn partitioned into $\{\bm{D}^i\}_{i\in \bm{N}}$ and $\{\bm{U}^i\}_{i\in \bm{N}}$ corresponding to their association with a particular agent $i \in \bm{N}$.
    \end{itemize}
\end{definition}

\begin{definition}
    \label{def:MAIM}
    A \textbf{multi-agent influence model (MAIM)} is a tuple $(\bm{N}, \bm{V}, \bm{E}, \theta)$ where $(\bm{N}, \bm{V}, \bm{E})$ is a MAID and:
    \begin{itemize}
        \item $\theta \in \Theta$ is a particular parametrisation over the nodes in the graph specifying a finite domain $\dom(V)$ for each node $V \in \bm{V}$, real-valued domains $\dom(U) \subset \mathbb{R}$ for each $U \in \bm{U}$, and a set of conditional probability distributions (CPDs) $\Pr(\bm{V} \mid \Pa_V)$ for every chance and utility node. Taken together, the CPDs form a partial distribution $\Pr(\bm{X},\bm{U} : \bm{D}) = \prod_{V \in \bm{V} \setminus \bm{D}} \Pr(V \mid \Pa_V)$ over the variables in the MAID.
        \item The value $u \in \dom(U)$ of a utility node is a deterministic function of the values of its parents $\pa_U \in \dom(\Pa_U)$. 
    \end{itemize}
\end{definition}

Figure \ref{fig:signalMACID} a) shows the MAID for Example \ref{ex:hiring} corresponding to the EFG in Figure \ref{fig:signalEFG}. Whether the worker is hard-working or lazy is decided by nature's chance node $X$ (white circle). The worker's decision $D^1$ and utility $U^1$ nodes are depicted as a red rectangle and diamond respectively. The algorithm's decision $D^2$ and utility $U^2$ nodes are in blue. To instantiate a MAIM, CPD tables for $U^1$ and $U^2$ would be consistent with the payoffs and value of $p$ in Figure \ref{fig:signalEFG}. 

There are two types of directed edge in a MAID. Full edges leading into $\bm{X} \cup \bm{U}$ represent probabilistic dependence, as in a Bayesian network.
Dotted edges leading into $\bm{D}$ represent information that is available to the agent at the time a decision is made (e.g. the edge $X \rightarrow D^1$). In this way, the values of the parents $\pa_D$ of a decision node $D$ represent the decision context for $D$.
The CPDs of decision nodes are not defined when a MAIM is constructed because they are instead chosen by the agents playing the game. In general, a player's decision CPD need not be optimal.

\begin{figure}[ht]
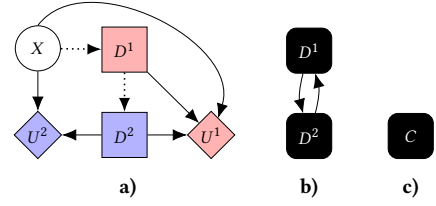

    \centering
    \begin{subfigure}[b]{0.4\linewidth}
        \centering
        \resizebox{0.75\width}{!}{
        \begin{influence-diagram}
        \node (X) [] {$X$};
        \node (U2) [utility, below = of X, player2] {$U^2$};
        \node (D2) [decision, right = of U2, player2] {$D^2$};
        \node (U1) [utility, right = of D2, player1] {$U^1$};
        \node (D1) [decision, above = of D2, player1] {$D^1$};
        \edge [information] {X} {D1};
        \edge [information] {D1} {D2};
        \path (X.north) edge[->, bend left=80] (U1);
        \edge {X} {U2};
        \edge {D1} {U1};
        \edge {D2} {U1, U2};
        \end{influence-diagram}
        }
        \caption*{a)}
    \end{subfigure}
    \begin{subfigure}[b]{0.15\linewidth}
        \centering
        \resizebox{0.75\width}{!}{
        \begin{influence-diagram}
        \node (U) [relevanceb] {$D^1$};
        \node (H) [below = of U, relevanceb] {$D^2$};
        \path (U) edge[->, bend right=15] (H);
        \path (H) edge[->, bend right=15] (U);
        \end{influence-diagram}
        }
        \caption*{b)}
    \end{subfigure}
    \begin{subfigure}[b]{0.15\linewidth}
        \centering
        \resizebox{0.75\width}{!}{
        \begin{influence-diagram}
        \node (U) [relevanceb] {$C$};
        \end{influence-diagram}
        }
        \caption*{c)}
    \end{subfigure}
 \caption{ A MAID $\macid$  a) representation of Example \ref{ex:hiring}, along with its cyclic relevance graph $Rel(\macid)$ b) (Section \ref{sec:stratrel}) and condensed relevance graph $ConRel(\macid)$ c) (Section \ref{sec:subgames}).}
    \label{fig:signalMACID}
\end{figure}

This example demonstrates two clear advantages of MAIDs compared with EFGs. First, in many real world cases, MAIDs make it possible to explicitly represent aspects of game structure that are obscured in the extensive form. For example, in the EFG, information sets were drawn to reflect the fact that the algorithm does not know whether the worker is naturally hard-working or lazy when it selects its action. However, in the corresponding MAID, this incomplete information is represented simply by the fact that there is no edge $X \rightarrow D^2$. Moreover, the company's utility $U^2$ isn't a function of whether the applicant went to university or not -- it only cares whether the applicant is hard-working and whether or not they hired them. We can infer this from the EFG payoffs in Figure \ref{fig:signalEFG}, but in the MAID this is shown instantly by the fact that there is no edge $D^1 \rightarrow U^2$.  

Second, MAIDs can provide a more compact graphical representation of games \cite{koller2003multi, milch2008ignorable}. In fact, the MAID representation of a game need never be bigger than the corresponding EFG and can be smaller in many cases. For example, there are four nodes $V^2_1, V^2_2, V^2_3, V^2_4$ in Figure \ref{fig:signalEFG} which correspond to the company's decision. In the MAID, these are combined into one node, $D^2$.

A further strength of MAIMs derive from them being probabilistic graphical models, and so probabilistic dependencies between chance and strategic variables can be exploited. We recall the notion of \textit{d-separation}, a graphical criterion for determining independence properties of the probability distribution associated with the graph. This is necessary for the concept of r-reachability (Section \ref{sec:stratrel}) and consequently that of a MAID subgame (Section \ref{sec:subgames}).

\begin{definition}[\cite{pearl2009causality}]
\label{def:dsep}
    A path $p$ in a MAID $(\bm{N}, \bm{V}, \bm{E})$ is said to be \textbf{d-separated} by a set of nodes $\bm{W} \subset \bm{V}$ if and only if either:
    \begin{itemize}
        \item $p$ contains a chain $X \rightarrow Y \rightarrow Z$ or a fork $X \leftarrow Y \rightarrow Z$ and $Y \in \bm{W}$.
        \item $p$ contains a collider $X \rightarrow Y \leftarrow Z$ and $(\{Y\} \cup \Desc_Y) \not\subseteq \bm{W}$.
    \end{itemize}
\end{definition}

A set $\bm{W}$ \textbf{d-separates} $\bm{X}$ from $\bm{Y}$, denoted $\bm{X} \perp \bm{Y} \mid \bm{W}$, if and only if $\bm{W}$ d-separates every path from a node in $\bm{X}$ to a node in $\bm{Y}$.
Sets of variables that are not d-separated are said to be \textbf{d-connected}, denoted $\bm{X} \not\perp \bm{Y} \mid \bm{W}$. 
If $\bm{X}$ and $\bm{Y}$ are d-separated conditioning on $\bm{W}$, then $\bm{X}$ and $\bm{Y}$ are probabilistically independent in the sense that $P(\bm{X}\mid \bm{Y}, \bm{W}) = P(\bm{X}\mid \bm{W})$.

For example, there are several paths from $U^2$ to $U^1$ in Figure \ref{fig:signalMACID} a): direct forks through $X$ or $D^2$, a fork through $X$ and then a forward chain through $D^1$, or a backward chain through $D^2$ and then a fork through $D^1$. If $\bm{W} = \varnothing$ then $U^2$ is d-connected to $U^1$ ($U^2 \not\perp U^1\mid\varnothing$), but if $\bm{W} = \{X, D^2\}$ then all of the paths have been d-separated by conditioning on $\bm{W}$ and so $U^2 \perp U^1\mid\bm{W}$.

\subsection{Policies}
\label{sec:policies}
An agent makes a decision depending on the information it observes prior to making that decision. Therefore, in a MAIM, a \textbf{decision rule} $\pi_D$ for a decision node $D$ is a CPD $\pi_D(D\mid\Pa_D)$. A \textbf{partial policy profile} $\pi_{\bm{A}}$ is an assignment of decision rules $\pi_D(D\mid\Pa_D)$ to some subset $\bm{A} \subset \bm{D}$ and $\pi_{-\bm{A}}$ is the set of decision rules for all $D \in \bm{D} \setminus \bm{A}$. For example, $\pi_{D}$ refers to a decision rule at decision node $D$ and so $\pi_{-D} = \prod_{D' \in \bm{D} \setminus \{D\}} \pi_{D'}(D'\mid\Pa_{D'})$ denotes the partial policy profile over all of the MAIM's other decision nodes $\bm{D} \setminus \{D\}$. We refer to $\pi_{\bm{D}^i}$, which describes all the decision choices made by agent $i \in \bm{N}$, as that agent's \textbf{policy}, $\pi^i$, and we write $\pi^{-i} = (\pi^1, \dots, \pi^{i-1}, \pi^{i+1}, \dots, \pi^n)$ to denote the set of policies made by all agents other than agent $i$. A \textbf{policy profile} $\pi$ assigns a policy to every agent $\pi = (\pi^1, \dots, \pi^n)$; it describes all the decisions made by every agent in the MAIM. We denote spaces of policy profiles by $\Pi$ (e.g.\ $\Pi_{\bm{A}}$, $\Pi^i$, and $\Pi$).

If for every $\pa_d \in \dom(\Pa_D)$ and $d \in \dom(D)$ we have $\pi_D(d\mid\pa_D) \in \{0,1\}$, the decision rule is said to be \textbf{pure} or \textbf{deterministic}. Otherwise, the decision rule is said to be \textbf{mixed} and it is \textbf{fully mixed} if, for every $\pa_D$ and every $d$, we have $\pi_D(d\mid\pa_D) > 0$. Pure, mixed, and fully mixed policies or policy profiles are defined analogously. 

When a partial policy profile $\pi_{\bm{A}}$ is applied to a MAIM $\macid$, a new MAIM $\macid(\pi_{\bm{A}})$ is obtained in which each decision node $D \in \bm{A}$ becomes a chance node with a CPD equal to $\pi_D$. In the case of a policy profile, all decision nodes are turned into chance nodes, and so the induced MAIM $\mathcal{M(\pi)}$ is now a Bayesian network (utility nodes are interpreted as chance nodes when a MAIM is viewed as a Bayesian network). This defines the joint probability distribution $\Pr^\pi$ over all variables in $\mathcal{M}$ and may be used for probabilistic inference.

\subsection{Utilities} \label{sec:reward}

In an EFG $\efg$ the expected utility for each player depends on the set of probability distributions $\bm{P}$ and strategy profile $\sigma$ which give a full probability distribution ${P}^\sigma$ over the paths in $\efg$. For each path $\rho$ beginning from the root $R$ of $\efg$'s tree and terminating in a unique leaf node $\rho[\bm{L}]$, player $i$ receives utility $U(\rho[\bm{L}])[i]$ -- the $i$\textsuperscript{th} entry in the corresponding payoff vector. By playing strategy profile $\sigma$, player $i$'s expected utility $\mathcal{U}^i_{\efg}(\sigma) \coloneqq \sum_\rho P^\sigma(\rho) U(\rho[\bm{L}])[i]$.

Similarly, the joint distribution $\Pr^\pi$ induced by the policy profile $\pi$ in a MAIM $\macid$ allows us to define the expected utility for each player under this policy profile. Agent $i$'s expected utility from policy profile $\pi$ is the sum of the expected value of utility nodes $\bm{U}^i$ given by
$\mathcal{U}^i_{\mathcal{M}}(\pi) 
 \coloneqq \sum_{U_j \in \bm{U}^i}\sum_{u_j \in \dom(U_j)} \!\! u_j \Pr^\pi(U_j = u_j)$.
 We assume that each agent's goal is to select a policy $\pi^i$ that maximises its expected utility. Therefore, we can now define what it means for an agent to optimise $\pi_{\bm{A}}$ for a set of decisions $\bm{A} \subseteq \bm{D}^i$, given a partial policy profile $\pi_{-\bm{A}}$ over all of the other decision nodes in $\macid$. We write $\mathcal{U}^i_{\mathcal{M}}(\pi_{\bm{A}}, \pi_{\bm{-A}})$ to denote the expected utility for player $i$ under the policy profile $\pi = (\pi_{\bm{A}}, \pi_{\bm{-A}})$.

\begin{definition} Let $\bm{A} \subseteq \bm{D}^i$. Player $i$'s partial policy $\pi_{\bm{A}}$ is \textbf{optimal} for a policy profile $\pi=(\pi_{\bm{A}}, \pi_{-\bm{A}})$ if $ \mathcal{U}^i_\mathcal{M}(\pi_{\bm{A}}, \pi_{-\bm{A}}) \geq 
    \mathcal{U}^i_\mathcal{M}(\hat{\pi}_{\bm{A}}, \pi_{-\bm{A}})$ for all $\hat{\pi}_{\bm{A}} \in \Pi_{\bm{A}}$.
Player $i$'s policy $\pi^{i}$ is a \textbf{best response} to the partial policy profile $\pi^{-i}$ assigning policies to the other agents if $\mathcal{U}^i_\mathcal{M}(\pi^{i}, \pi^{-i}) \geq 
    \mathcal{U}^i_\mathcal{M}(\hat{\pi}^i, \pi^{-i})$ for all $\hat{\pi}^i \in \Pi^i$.
\label{def:bestresponse}
\end{definition}

\subsection{Strategic and Probabilistic Relevance}
\label{sec:stratrel}

 To optimise a particular decision rule, we often want to know which other decision rules need to already be known. This is captured by K$\&$M's concept of \textit{strategic relevance}.

\begin{definition}[\cite{koller2003multi}]
    \label{def:stratrel}
    Let $D_k, D_l \in \bm{D}$ be decision nodes in a MAIM $\mathcal{M}$. $D_l$ is \textbf{strategically relevant} to $D_k$ ($D_k$ strategically relies on $D_l$) if there exist two policy profiles $\pi$ and $\pi'$ and a decision rule $\pi_{D_k}$, such that:
    \begin{itemize}
        \item $\pi_{D_k}$ is optimal for $\pi$.
        \item $\pi$ differs from $\pi'$ only at $D_l$.  
        \item $\pi_{D_k}$ is not optimal for $\pi'$, and neither is any decision rule $\hat{\pi}_{D_k}$ that agrees with $\pi_{D_k}$ for all instantiations $\pa_{D_k}$ of $D_k$'s parents where the joint probability $\Pr^{\pi'}(\pa_{D_k}) > 0$.
    \end{itemize}
\end{definition}

The first two conditions say that if decision rule $\pi_{D_k}$ is optimal for a policy profile $\pi$, and $D_k$ does not strategically rely on $D_l$, then $\pi_{D_k}$ must also be optimal for any policy profile $\pi'$ that differs from $\pi$ only at $D_l$. The third condition deals with sub-optimal decisions in response to zero-probability decision contexts. 

A related question is \textit{probabilistic relevance}, which considers whether the probability distribution of a chance or utility node $X$ can influence the optimal policy.


\begin{definition}
    Let $D$ be a decision node in a MAID $\macid$. 
    A chance or utility node $Z \in \bm{X} \cup \bm{U}$ is \textbf{probabilistically relevant} to $D$ if the set of optimal decision rules for $D$ varies with the CPDs assigned to $Z$ under some joint policy profile $\pi$.
\end{definition}

We generalise K\&M's graphical criterion, \textbf{s-reachability}, as \textbf{r-reachability} to determine both strategic relevance and probabilistic relevance.
Essentially, the criterion assesses whether knowing the CPD or decision rule of a node $V$ can have positive value of information \cite{Everitt2021}.
The criterion is sound (if $V$ is relevant to $D$, then $V$ is r-reachable from $D$) and complete (if $V$ is r-reachable from $D$ then there is some parametrisation $\theta$ of the MAID and some policy profile $\pi$ such that $V$ is relevant to $D$). One can then further use r-reachability to define a \textit{relevance graph} over $\bm{D}$.

\begin{definition} 
\label{def:rreachable}
    A node $V$ is \textbf{r-reachable} from a decision $D \in \bm{D}^i$ in a MAID, $\macid = (\bm{N}, \bm{V}, \bm{E})$, if a newly added parent $\hat V$ of $V$ satisfies $\hat V \not\perp \bm{U}^i \cap \Desc_{D} \mid \Fa_{D}$. 
\end{definition}

\begin{definition}\label{def:relgraph}
    The directed \textbf{relevance graph} for $\mathcal{M}$, denoted by $Rel(\mathcal{M}) = (\bm{D}, \bm{E}_{Rel})$, is a graph where $\bm{D}$ is the set of $\mathcal{M}$'s decision nodes connected by directed edges $\bm{E}_{Rel} \subseteq \bm{D} \times \bm{D}$. There is a directed edge from $D_j \rightarrow D_k$ if and only if $D_k$ is r-reachable from $D_j$.\footnote{The edge directions used here are the same as originally defined by K\&M \cite{koller2001multi} but reversed compared with those in their later work \cite{koller2003multi} as this eases our later exposition of MAID subgames (Section \ref{sec:subgames}).}
\end{definition}

Relevance graphs show which other decisions each decision depends on. The relevance graph for Example \ref{ex:hiring}'s MAIM in Figure \ref{fig:signalMACID} b) is cyclic because each decision node strategically relies on the other. The worker would be better off knowing the company's hiring policy before deciding whether or not to go to university, but the algorithm would also be better off knowing the worker's policy because it doesn't know the worker's temperament (lazy or hard-working). Our second example provides a case of acyclic strategic relevance.

\begin{example}[Taxi competition]
\label{ex:taxis}
\textit{Two autonomous taxis, operated by different companies, are driving along a road with two hotels located next to one another -- one expensive and one cheap. Each taxi must decide (one first, then the other) which hotel to stop in front of, knowing that it will likely receive a higher tip from guests of the expensive hotel. However, if both taxis choose the same location, this will reduce each taxi's chance of being chosen by that hotel's guests.}
\end{example}

\begin{figure}[ht!]
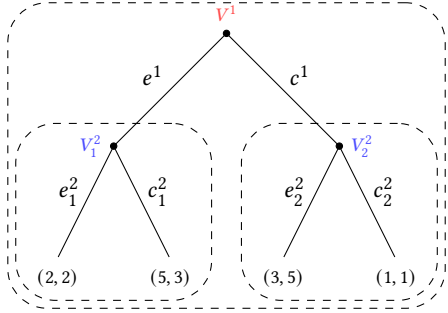

    \centering
    \begin{istgame}[scale=0.75]
    \xtdistance{20mm}{20mm}
    \istroot(0)<90,red!70>{$V^1$}+20mm..40mm+
    \istb{e^1}[al]
    \istb{c^1}[ar] \endist
    \istroot(1)(0-1)<180, blue!70>{$V^2_1$}
    \istb{e^2_1}[al]{(2,2)}
    \istb{c^2_1}[ar]{(5,3)} \endist
    \istroot(2)(0-2)<0, blue!70>{$V^2_2$}
    \istb{e^2_2}[al]{(3,5)}
    \istb{c^2_2}[ar]{(1,1)} \endist
    \xtSubgameBox(1){(1-1)(1-2)}[black,inner sep = 15pt]
    \xtSubgameBox(2){(2-1)(2-2)}[black,inner sep = 15pt]
    \xtSubgameBox(0){(1-1)(1-2)(2-1)(2-2)}[black,inner sep = 18pt]
    \end{istgame}
    \caption{An EFG representation of Example \ref{ex:taxis}, with EFG subgames enclosed in dashed boxes.}
    \label{fig:taxiEFG}
\end{figure}

Because the second taxi can observe which hotel the first taxi chooses to park in front of, it doesn't need to know the first taxi's policy in order to optimise its own; the second taxi's decision ($D^2$) does \emph{not} strategically rely on the first taxi's decision ($D^1$). However, the first taxi would be better off knowing the second taxi's policy before deciding its own. For example, with the parametrisation in Figure \ref{fig:taxiMACID}  e), if the first taxi knows that the second taxi's policy is to always park in front of the expensive hotel, the first taxi ought to always park in front of the cheaper hotel. $D^1$ \emph{does} strategically rely on $D^2$. In Section \ref{sec:implementation}, we shall see that it is easier to compute equilibria in MAIMs with acyclic relevance graphs.

\begin{figure}[h]
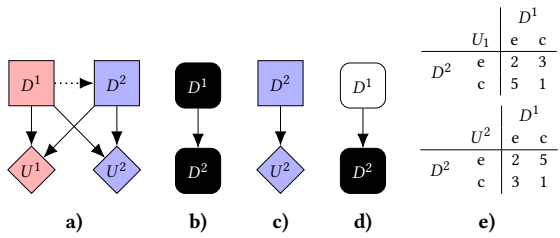

    \centering
    \begin{subfigure}[b]{0.25\linewidth}
        \centering
        \resizebox{0.75\width}{!}{
        \begin{influence-diagram}
            \node (D1) [decision, player1] {$D^1$};
            \node (D2) [decision, right = of D1, player2] {$D^2$};
            \node (U1) [utility, below = of D1, player1] {$U^1$};
            \node (U2) [utility, below = of D2, player2] {$U^2$};
            \edge [information] {D1} {D2};
            \edge {D1, D2} {U1};
            \edge {D1, D2} {U2};
        \end{influence-diagram}
        }
        \caption*{a)}
    \end{subfigure}
    \begin{subfigure}[b]{0.12\linewidth}
        \centering
        \resizebox{0.75\width}{!} {
        \begin{influence-diagram}
            \node (D1) [right = 2 of U, relevanceb] {$D^1$};
            \node (D2) [below = of D1, relevanceb] {$D^2$};
            \edge {D1} {D2};
        \end{influence-diagram}
        }
        \caption*{b)}
    \end{subfigure}
    \begin{subfigure}[b]{0.12\linewidth}
        \centering
        \resizebox{0.75\width}{!}{
        \begin{influence-diagram}
            \node (D2) [decision, player2] {$D^2$};
            \node (U2) [utility, below = of D2, player2] {$U^2$};
            \edge {D2} {U2};
        \end{influence-diagram}
        }
        \caption*{c)}
    \end{subfigure}
    \begin{subfigure}[b]{0.12\linewidth}
        \centering
        \resizebox{0.75\width}{!}{
        \begin{influence-diagram}
            \node (D1) [right = 2 of U, relevancew] {$D^1$};
            \node (D2) [below = of D1, relevanceb] {$D^2$};
            \edge {D1} {D2};
        \end{influence-diagram}
        }
        \caption*{d)}
    \end{subfigure}
    \begin{subfigure}[b]{0.25\linewidth}
        \centering
        \resizebox{0.75\width}{!}{
            \begin{tabular}{cc|cc}
            &    & \multicolumn{2}{c}{$D^1$}   \\
            & $U_1$ & e        & c \\ \hline
    \multirow{2}{*}{$D^2$} & e  & 2        & 3 \\
            & c  & 5        & 1 \\
            \end{tabular}}
            \\
            \vspace{0.1cm}
            \resizebox{0.75\width}{!}{
            \begin{tabular}{cc|cc}
            &    & \multicolumn{2}{c}{$D^1$}   \\
            & $U^2$ & e        & c \\ \hline
    \multirow{2}{*}{$D^2$} & e  & 2        & 5 \\
            & c  & 3       & 1 \\
            \end{tabular}}
            \caption*{e)}
\end{subfigure}
\caption{
A MAID a) and corresponding relevance graph b) for Example \ref{ex:taxis}, alongside the only proper MAID subgame c) highlighted in the (condensed) relevance graph d). The utility nodes' parametrisation is in e).}
\label{fig:taxiMACID}
\end{figure}

\section{Game Theory for MAIDs}

In this section, we present novel material. We begin by defining MAID and MAIM subgames. These set up our discussion of several equilibrium refinements in MAIMs. Finally, we present and prove a number of equivalence results between EFGs and MAIMs.

\subsection{Subgames}
\label{sec:subgames}

In an EFG, a subtree of the original game tree is an \textbf{EFG subgame} if it is closed under information sets and descendants. Figure \ref{fig:taxiEFG} shows all the EFG subgames (dashed boxes) for the game described by Example \ref{ex:taxis}. Any game tree is an EFG subgame of itself, and so an EFG subgame on a strictly smaller set of nodes is called a \textbf{proper} EFG subgame. We propose an analoguous definition for MAIDs. Just like for EFGs, MAIM subgames are parts of the game that can be solved independently.

\begin{definition}
    A \textbf{subgame base} for a MAID $(\bm{N}, \bm{V}, \bm{E})$ is a subset $\bm{V'}\subseteq \bm{V}$ such that:
    \begin{itemize}
        \item For any $X,Y\in \bm{V}'$ and any directed path $X \to \cdots \to Y$ in $\macid$, all nodes on the path are also in $\bm{V'}$.
        \item $\bm{V}'$ is closed under r-reachability, i.e.\ if a node $Z$ is r-reachable from a decision $D\in \bm{V'}$, then $Z$ is also in $\bm{V'}$.
    \end{itemize}
\end{definition}

\begin{definition}
\label{def:MACIDsubgame}
    Let $\macid = (\bm{N}, \bm{V}, \bm{E})$ be a MAID,
    and let $\bm{V'}\subseteq \bm{V}$ be a subgame base.
    The \textbf{MAID subgame} corresponding to $\bm{V'}$, is a new MAID $\macid'=(\bm{N'}, \bm{V'}, \bm{E}')$  where:
    \begin{itemize}
        \item $\bm{N'} = \{i \in \bm{N} \mid \bm{D}^i \cap \bm{V'} \neq \varnothing \}$, the players restricted to $\bm{V'}$.
        \item $\bm{V'}$ is partitioned into $\bm{D'} = \bm{D}\cap \bm{V'}$, $\bm{U'} = \bm{U} \cap \Desc_{\bm{D'}}$, and $\bm{X'} = \bm{V'} \setminus (\bm{D'} \cap \bm{U'})$.
        \item $\bm{E}'$ is the subset of edges in $\bm{E}$ that connect two nodes in $\bm{V'}$.
    \end{itemize}
    Analogously, the \textbf{MAIM subgame}
    of a MAIM $(\bm{N}, \bm{V}, \bm{E}, \theta)$ corresponding to a subset $\bm{V'}\subseteq \bm{V}$ and an instantiation $\bm{y}$ of the nodes $\bm{Y} = \bm{V}\setminus \bm{V'}$,
    is the modified MAIM $(\bm{N'}, \bm{V'}, \bm{E}', \theta')$ where:
    \begin{itemize}
        \item $(\bm{N'}, \bm{V'}, \bm{E}')$ is the MAID subgame corresponding to $\bm{V'}$.
        \item $\theta'$ is like $\theta$, restricted to nodes in $\bm{V'}$.
        If a node $X\in \bm{V'}$ has some parents outside of $\bm{V'}$ (i.e.\ in $\bm{Y}$) then $\Pr'(X\mid \pa'_X) = \Pr(X\mid \pa'_X, \bm{y}')$, where $\Pa'_X=\Pa_X\cap \bm{V'}$, $\bm{Y}' = \Pa_X\cap \bm{Y}$, $\Pr$ is the CPD of $X$ in $\theta$, and $\Pr'$ becomes the CPD of $X$ in $\theta'$.
    \end{itemize}
    In fact, only the setting $\bm{y}$ of the nodes that have a child in $\bm{V'}$ will matter.
    A MAIM subgame is \textbf{feasible} if
    there exists a policy profile $\pi$ where $\Pr^\pi(\bm{y}) > 0$.
\end{definition}

In a sequential game with perfect information, the MAIM subgames will be in one-to-one correspondence with the subgames in any corresponding EFG. For example, Figures \ref{fig:taxiEFG} and \ref{fig:taxiMACID} show the EFG and MAID subgames of Example \ref{ex:taxis}. As with EFG subgames, a MAID is trivially a MAID subgame of itself, as in Figure \ref{fig:taxiMACID} a). Figure \ref{fig:taxiMACID} c) shows the only \textit{proper} MAID subgame of $\macid$. 
Two MAIM subgames are associated with this MAIM subgame: one for each value of $D^1$.
The additional independencies represented by a MAID sometimes yields more independently solvable components than identifiable in an EFG representation; i.e., there can be more subgames in a MAIM than in a corresponding EFG (Appendix \ref{sec:extra_subgames}).

A better sense of MAID subgames can be gained from looking at the strongly connected components (SCC) of the relevance graph, where recall that an SCC is a subgraph containing a directed path between every pair of nodes. A maximal SCC is an SCC that is not a strict subset of any other SCC. We can use this fact to define a condensed relevance graph, called the component graph by K$\&$M, which aggregates each SCC into a single node.

\begin{definition} \label{def:ConRel} For a given MAID $\mathcal{M} = (\bm{N}, \bm{V}, \bm{E})$, let $\bm{C}$ be the set of maximal SCCs of its relevance graph $Rel(\macid)$.
The \textbf{condensed relevance graph} of $\mathcal{M}$ is the directed graph $ConRel(\macid) = (\bm{C}, \bm{E}_{ConRel})$. There is an edge $\bm{C}_m \rightarrow \bm{C}_l$ between $\bm{C}_m, \bm{C}_l\in\bm{C}$ if and only if there exists $C_m\in \bm{C}_m$ and $C_l\in \bm{C}_l$ with and edge $C_m\to C_l$ in $Rel(\macid)$.
\end{definition}

Subgraphs of $ConRel(\macid)$ closed under descendants induce MAID subgames.
Figures \ref{fig:taxiMACID} b) and d) highlight the nodes of the respective MAID subgames (since the relevance graph is acyclic here, condensing it has no effect).
As the condenstion of a directed graph is always acyclic \cite[page 617]{cormen2009introduction}, games can always be solved via backwards induction over the condensed relevance graph \citep{koller2003multi}.

\subsection{Equilibrium Refinements}
\label{sec:eqrefine}

MAIDs represent dynamic games of incomplete information -- those in which at least one player $i$ does not have perfect information about the chance variables $\bm{X}$ (commonly interpreted as not knowing the \emph{type} $T^i$ of the other agents, where $T^i$ defines the payoffs for agent $i$) -- and thus admit discussion of the \emph{beliefs} that agents possess. In this work, we implicitly view such beliefs as defined by the induced distribution $\Pr^\pi$ and so eschew further discussion of them here; in a sense, chance variables can be viewed as decisions by nature (using fixed stochastic policies). 

In non-cooperative games, the most fundamental solution is a Nash equilibrium \cite{nash1950equilibrium}, a policy profile such that no agent has an incentive to unilaterally deviate. In other words, every player is simultaneously playing a best-response against all other players.

\begin{definition}[\cite{koller2003multi}] 
    \label{def:NE}
    A full policy profile $\pi$ is a \textbf{Nash equilibrium (NE)} in a MAIM $\mathcal{M}$ if, for every player $i \in \bm{N}$, $\mathcal{U}^i_\mathcal{M}(\pi^i, \pi^{-i}) \geq \mathcal{U}^i_\mathcal{M}(\hat{\pi}^i, \pi^{-i})$ for all $\hat{\pi}^i \in \Pi^i$.
\end{definition}

The concept of a subgame perfect equilibrium (SPE) was introduced by Reinhard Selten to address the issue that EFGs admit NEs with \textit{non-credible threats} -- equilibria in which a player threatens to take some action that, if the player is rational, they would never actually carry out \cite{selten1965spieltheoretische, selten1974reexamination}. In an EFG, a strategy profile is an SPE if it induces an NE in every EFG subgame; this eliminates all NEs containing non-credible threats. Our definition of MAID subgames above allows us to introduce an analogous equilibrium concept.

\begin{definition}
    A full policy profile $\pi$ is a \textbf{subgame perfect equilibrium (SPE)} in a MAIM $\macid$ if $\pi$ is an NE in every MAIM subgame of $\macid$.
\end{definition}

\begin{figure}[h]
    \centering
    \begin{subfigure}[c]{0.3\linewidth}
        \centering
        \resizebox{0.75\width}{!}{
            \begin{tabular}{c | c}
                $D^1$ & \\
                \hline
                $e$ & 1 \\
                $c$ & 0 \\
            \end{tabular}}
            \\
            \vspace{0.1cm}
            \resizebox{0.75\width}{!}{
            \begin{tabular}{cc|cc}
            &    & \multicolumn{2}{c}{$D^1$}   \\
            &    & e        & c \\ \hline
    \multirow{2}{*}{$D^2$} & e  & 0        & 1 \\
            & c  & 1        & 0
            \end{tabular}}
            \caption*{a)}
    \end{subfigure}
    \begin{subfigure}[c]{0.3\linewidth}
        \centering
        \resizebox{0.75\width}{!}{
            \begin{tabular}{c | c}
                $D^1$ & \\
                \hline
                $e$ & 0 \\
                $c$ & 1 \\
            \end{tabular}}
            \\
            \vspace{0.1cm}
            \resizebox{0.75\width}{!}{
            \begin{tabular}{cc|cc}
            &    & \multicolumn{2}{c}{$D^1$}   \\
            &    & e        & c \\ \hline
    \multirow{2}{*}{$D^2$} & e  & 1        & 1 \\
            & c  & 0        & 0
            \end{tabular}}
            \caption*{b)}
    \end{subfigure}
    \begin{subfigure}[c]{0.3\linewidth}
        \centering
        \resizebox{0.75\width}{!}{
            \begin{tabular}{c | c}
                $D^1$ & \\
                \hline
                $e$ & 1 \\
                $c$ & 0 \\
            \end{tabular}}
            \\
            \vspace{0.1cm}
            \resizebox{0.75\width}{!}{
            \begin{tabular}{cc|cc}
            &    & \multicolumn{2}{c}{$D^1$}   \\
            &    & e        & c \\ \hline
    \multirow{2}{*}{$D^2$} & e  & 0        &  0\\
            & c  & 1        & 1
            \end{tabular}}
            \caption*{c)}
    \end{subfigure}
    \caption{The policies for Example \ref{ex:taxis}'s three pure NEs. Only a) is an SPE.}
    \label{fig:taxiNE}
\end{figure}

Figure \ref{fig:taxiNE} shows the three pure NEs of Example \ref{ex:taxis}.
The policy profiles in b) and c) are NEs but not SPEs. To see why, consider the proper MAIM subgame when $D^1 = e$ and policy profile b). 
Here player 2 obtains utility 3 if they choose $c$ and utility 2 if they choose $e$. Therefore, player 2 is making a \textit{non-credible} threat whenever  
$\pi^2(D^2=e \mid D^1=e) > 0$.
For similar reasons, policy profile c) is also not SPE.
Therefore a) is the only SPE of this MAIM.

Within maximal SCCs of $Rel(\macid)$, in which there are no proper subgames, the agents choose decision rules interdependently. This can lead to arbitrarily bad decision rules in decision contexts that occur with probability zero. 
Trembling hand equilibria offer a useful NE refinement in these situations \cite{selten1974reexamination}. Intuitively, they require that each player's policy is still a best response when the other players make mistakes, or `tremble', with small probability.
Let $\delta_k$ be a perturbation vector containing, for every $D \in \bm{D}$, $d \in \dom(D)$, and decision context $\pa_D$, a value $\epsilon^d_{\pa_D} \in (0,1)$ such that $\sum_{d \in \dom(D)} \epsilon^d_{\pa_D} \leq 1$. Then, given a MAIM $\macid$, the perturbed MAIM $\macid(\delta_k)$ is defined such that for every $d \in \dom(D)$ for $D \in \bm{D}^i$, agent $i$ must play $d$ with probability at least $\epsilon^d_{\pa_D}$ given $\pa_D$.

\begin{definition}
    A full policy profile $\pi$ is a \textbf{trembling hand perfect equilibrium (THPE)} in a MAIM $\macid$ if there is a sequence of perturbation vectors $\{\delta_k\}_{k\in\mathbb{N}}$ such that $\lim_{k \rightarrow \infty}\vert\delta_k\vert_\infty = 0$ and for each perturbed MAIM $\macid(\delta_k)$ there is an NE $\pi_k$ such that $\lim_{k \rightarrow \infty} \pi_k = \pi$.
\end{definition}

\begin{example}[Cyber-war] 
\label{ex:cyber}
\textit{The security agencies for two governments both use an algorithm to manage their cyber-defence. Their algorithm decides whether to cyber-attack the other nation's security agency. If both agencies attack one another, both suffer some damage (mainly the opportunity cost of needing to continuously work on upgrading their defence systems). The attacker never gains much, but if only one agency attacks the other, the defender suffers a lot more damage.}
\end{example}

\begin{figure}[h]
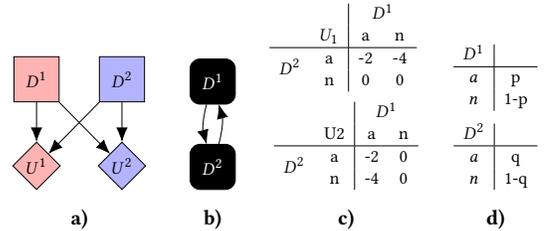

    \centering
    \begin{subfigure}[b]{0.3\linewidth}
        \centering
        \resizebox{0.75\width}{!}{
        \begin{influence-diagram}
            \node (D1) [decision, player1] {$D^1$};
            \node (D2) [decision, right = of D1, player2] {$D^2$};
            \node (U1) [utility, below = of D1, player1] {$U^1$};
            \node (U2) [utility, below = of D2, player2] {$U^2$};
            \edge {D1, D2} {U1};
            \edge {D1, D2} {U2};
        \end{influence-diagram}
        }
        \caption*{a)}
    \end{subfigure}
    \begin{subfigure}[b]{0.1\linewidth}
        \centering
        \resizebox{0.75\width}{!}{
        \begin{influence-diagram}
            \node (D1) [right = 2 of U, relevanceb] {$D^1$};
            \node (D2) [below = of D1, relevanceb] {$D^2$};
            \path (D1) edge[->, bend right=15] (D2);
            \path (D2) edge[->, bend right=15] (D1);
        \end{influence-diagram}
        }
        \caption*{b)}
    \end{subfigure}
    \begin{subfigure}[b]{0.3\linewidth}
        \centering
        \resizebox{0.75\width}{!}{
            \begin{tabular}{cc|cc}
            &    & \multicolumn{2}{c}{$D^1$}   \\
            & $U_1$ & a        & n \\ \hline
    \multirow{2}{*}{$D^2$} & a  & -2        & -4 \\
            & n  & 0        & 0 \\
            \end{tabular}
            }
            \vspace{0.1cm}
            \\
            \resizebox{0.75\width}{!}{
            \begin{tabular}{cc|cc}
            &    & \multicolumn{2}{c}{$D^1$}   \\
            & U2 & a        & n \\ \hline
    \multirow{2}{*}{$D^2$} & a  & -2        & 0 \\
            & n  & -4        & 0 \\
            \end{tabular}
            }
            \caption*{c)}
    \end{subfigure}
    \begin{subfigure}[b]{0.15\linewidth}
        \centering
            \resizebox{0.75\width}{!}{
            \begin{tabular}{c | c}
                $D^1$ & \\
                \hline
                $a$ & p \\
                $n$ & 1-p \\
            \end{tabular}
            }
            \vspace{0.1cm}
            \\
            \resizebox{0.75\width}{!}{
             \begin{tabular}{c | c}
                $D^2$ & \\
                \hline
                $a$ & q \\
                $n$ & 1-q \\
            \end{tabular}
            }
            \caption*{d)}
    \end{subfigure}
\caption{A MAID a), corresponding relevance graph b), utility CPD tables c), and policy profile d) for Example \ref{ex:cyber}. 
}
\label{fig:tremble}
\end{figure}

Figure \ref{fig:tremble} shows a MAIM with its parametrisation and relevance graph for Example \ref{ex:cyber}. Using each player's parametrised policy in Figure \ref{fig:tremble} d) this MAIM has two NEs: at $p=q=1$ and $p=q=0$, either both governments attack ($a$) or not ($n$). Figures \ref{fig:THpolicy} a) and c) show player $1$'s best response policies for each of these NEs perturbed by $\epsilon$ to result in the perturbed MAIM $\macid(\epsilon)$. Figures \ref{fig:THpolicy} b) and d) show player $2$'s expected utility if they attack (or not) in response to player $1$ using the policy in \ref{fig:THpolicy} a) or c) respectively. For small $\epsilon>0$, player $2$'s best response to both the policies in Figures \ref{fig:THpolicy} a) and c) is to choose $D^2 = a$
and so the NE $p=q=0$ is not robust against trembles. The NE $p=q=1$ is this MAIM's only THPE.\footnote{In the case of a two-player game, a THPE removes all weakly dominated policies. Here, for both $D^1$ and $D^2$, the pure policy of choosing $n$ is weakly dominated by the pure policy of choosing $a$.}

\begin{figure}[h]
    \centering
    \begin{subfigure}[b]{0.22\linewidth}
        \centering
        \resizebox{0.75\width}{!}{
            \begin{tabular}{c | c}
                $D^1$ & \\
                \hline
                $a$ & $1-\epsilon$ \\
                $n$ & $\epsilon$ \\
            \end{tabular}}
            \caption*{a)}
    \end{subfigure} 
    \begin{subfigure}[b]{0.26\linewidth}
        \centering
        \resizebox{0.75\width}{!}{
            \begin{tabular}{c | c}
                $D^2$ & $\mathcal{U}^2_{\mathcal{M}}(\pi)$\\
                \hline
                $a$ & $-2 + 2\epsilon$ \\
                $n$ & $-4+4\epsilon$ \\
            \end{tabular}}
            \caption*{b)}
    \end{subfigure} 
    \begin{subfigure}[b]{0.22\linewidth}
        \centering
        \resizebox{0.75\width}{!}{
            \begin{tabular}{c | c}
                $D^1$ & \\
                \hline
                $a$ & $\epsilon$ \\
                $n$ & $1-\epsilon$ \\
            \end{tabular}}
            \caption*{c)}
    \end{subfigure} 
    \begin{subfigure}[b]{0.26\linewidth}
        \centering
        \resizebox{0.75\width}{!}{
            \begin{tabular}{c | c}
                 $D^2$ & $\mathcal{U}^2_{\mathcal{M}}(\pi)$\\
                \hline
                $a$ & $-2\epsilon$ \\
                $n$ & $-4\epsilon$ \\
            \end{tabular}}
            \caption*{d)}
    \end{subfigure} 
    \caption{Polices a) and c) for player 1 for each NE in the original MAIM, shown in Figure \ref{fig:tremble}, perturbed by $\epsilon$, and the expected utilities b) and d) for player 2 when choosing each $d \in \dom(D^2)$ in response to policies a) and c) respectively.}
    \label{fig:THpolicy}
\end{figure}

\subsection{Transformations and Equivalences}
\label{sec:trans_equivalences}

Both EFGs and MAIMs represent games graphically. In this section, we provide equivalence results between these models to demonstrate that alongside the increased compactness and structural clarity of MAIMs, the fundamental game-theoretic notions of subgames and equilibria are \emph{preserved} when converting an EFG to a MAIM. 

\subsubsection{\textbf{MAIM to EFG}}
\label{sec:MAIDtoEFG}
There are many ways to convert a MAIM into an EFG, but these differ in their computational costs \cite{koller2003multi, pearl2014probabilistic}. We give a full and formal transformation procedure, \bfsf{maim2efg}, in Appendix \ref{sec:macid2efg} based on that of K\&M. The idea is to use a topological ordering $\prec$ over the nodes of the MAID to construct the EFG game tree by splitting on each of the nodes in $\prec$. Because there can be more than one such ordering, the output of \bfsf{maim2efg} is a \emph{set} of EFGs.
Our codebase implements a more efficient transformation, keeping only utility nodes, decision nodes, and informational parents ($ \bigcup_{D\in\bm{D}} \Fa_{D}$).
This information is enough for computing equilibria, and can offer significant efficiency gains since the cost of solving an EFG depends on its size, which is exponential in the length of $\prec$. The resulting EFG can be fed into Gambit, a popular tool for solving EFGs \cite{mckelvey2006gambit}, though it may not contain enough information to fully recover the original MAIM.

\subsubsection{\textbf{EFG to MAIM}}
\label{sec:EFGtoMAID}
By encoding the CPDs for each variable in the MAIM using trees as opposed to tables, MAIMs can represent any decision-making problem using at most the same space, but often exponentially less space than an EFG \cite{koller2003multi}. In general, there are many MAIMs that can represent a given EFG. For instance, upon converting the EFG representation (Figure \ref{fig:taxiEFG}) of Example \ref{ex:taxis} to a MAIM (Figure \ref{fig:taxiMACID}), we could na\"{i}vely retain the EFG's root and two child nodes as three decision nodes ($D^1$, $D^2_a$, and $D^2_b$) in the MAIM. Alternatively, we could recognise that $D^2_a$ and $D^2_b$ both correspond to the same real world variable, the decision made following $D^1$, and thus combine them (as shown in Figure \ref{fig:taxiMACID}). In Appendix \ref{sec:EFG2MAID}, we formalise this notion and provide a procedure \bfsf{efg2maim} which maps an EFG to a \emph{unique}, canonical MAIM (including those with absent-mindedness \cite{piccione1997interpretation}).

\subsubsection{\textbf{Equivalences}}
\label{sec:equivalences}

We now provide a series of equivalence results between EFGs and MAIMs to fortify the game-theoretic foundations behind our analysis of MAIDs. Results are justified using intuitive sketches, with full proofs in Appendix \ref{sec:proofs}.

\begin{definition}
\label{def:feasibleDC}
    A decision context $\pa_D$ for a decision node $D$ in $\macid$ if \textbf{feasible} if there exists a policy profile $\pi$ where $\Pr^\pi(\pa_D) > 0$. A decision context $\pa_D$ is \textbf{null} if every player  always receives utility 0, i.e. $\mathcal{U}^i_{\mathcal{M}}(\pi'' \mid \pa_D) = \sum_{U_j \in \bm{U}^i}\sum_{u_j \in \dom(U_j)} \!\! u_j \Pr^\pi(U_j = u_j \mid \pa_D) = 0$ for all $i$ and any policy profile $\pi''$, or if it is infeasible.
\end{definition}

\begin{definition}
    \label{def:equiv}
    We say that a MAIM $\macid$ is \textbf{equivalent} to an EFG $\efg$ (and vice versa) if there is a bijection $f: \Sigma \rightarrow \Pi / \sim$ between the strategies in $\efg$ and a partition of the policies in $\macid$ (the quotient set of $\Pi$ by an equivalence relation $\sim$) such that:
    \begin{itemize}
        \item $\pi\sim\pi'$ only if $\pi$ and $\pi'$ differ only on null decision contexts.
        \item For every $\pi \in f(\sigma)$ and every player $i$, $\mathcal{U}^i_{\efg}(\sigma) = \mathcal{U}^i_{\macid}(\pi)$.
    \end{itemize}
    We refer to $f$ as a \textbf{natural mapping} between $\efg$ and $\macid$.
\end{definition}

The reason we use an equivalence relation on the space of policies is that \bfsf{efg2maim} can introduce additional null decision contexts: those that do not correspond to any path through the EFG. Although this equivalence is not exact, it is sufficient for preserving the essential game-theoretic features of each representation, as we show below. We begin with a supporting lemma that justifies the correctness of our procedures \bfsf{maim2efg} and \bfsf{efg2maim}, and forms the basis of our other results.

\begin{lemma}
    \label{lem:correspondence}
    If $\efg\in\bfsf{maim2efg}(\macid)$ or
    $\macid =\bfsf{efg2maim}(\efg)$ then $\efg$ and $\macid$ are equivalent.
\end{lemma}

This lemma follows directly by construction from the two procedures, \bfsf{maim2efg} and \bfsf{efg2maim} respectively. The intuition is that the information sets in an EFG correspond to the non-null decision contexts in a MAIM, and thus an EFG's behavioural strategy profile $\sigma$ corresponds to a policy profile $\pi$ in the MAIM, and vice versa. As an immediate consequence, we see that NEs are preserved by our transformations between EFGs and MAIMs.

\begin{corollary}
\label{prop:BE}
    If $\efg\in\bfsf{maim2efg}(\macid)$ or $\macid = \bfsf{efg2maim}(\efg)$ then there is a natural mapping $f$ between $\efg$ and $\macid$ such that $\sigma$ is an NE in $\efg$ if and only if any $\pi \in f(\sigma)$ is an NE in $\macid$.
\end{corollary}

For an EFG subgame $\efg'$, the variables outside $\efg'$ are neither strategically nor probabilistically relevant to those in the corresponding MAIM subgame $\macid'$.
This means that EFG subgames have equivalent counterparts in the equivalent MAIM, as established by the following proposition.

\begin{proposition}
\label{prop:subgames}
    If $\efg\in\bfsf{maim2efg}(\macid)$ or $\macid = \bfsf{efg2maim}(\efg)$ then there is a natural mapping $f$ between $\efg$ and $\macid$ such that, for every EFG subgame $\efg'$ in $\efg$ there is a MAIM subgame $\macid'$ in $\macid$ that is equivalent to $\efg'$ under the natural mapping $f$ restricted to the strategies of $\efg'$.
\end{proposition}

This restriction of $f$ to the strategies in $\efg'$ can be made precise by considering only those non-null decision contexts that correspond to the information sets contained in $\efg'$, as in the case for Lemma \ref{lem:correspondence}. Given Proposition \ref{prop:subgames} and Corollary \ref{prop:BE}, it can easily be seen that not only are NEs preserved when representing EFGs as MAIMs, but so too are SPEs. We remark, however, that as there may be more subgames in MAIM than in an equivalent EFG, that the criterion of subgame perfectness may be slightly stronger in the MAIM, and so not all SPEs in an EFG may be SPEs in the equivalent MAIM. This additional strength can be useful in ruling out what we intuitively view as `irrational' behaviour, even when it does not fall under a particular subgame in the EFG.

\begin{corollary}
\label{prop:SPE}
    If $\efg\in\bfsf{maim2efg}(\macid)$ or $\macid = \bfsf{efg2maim}(\efg)$ then there is a natural mapping $f$ between $\efg$ and $\macid$ such that if any $\pi \in f(\sigma)$ is an SPE in $\macid$, then $\sigma$ is an SPE in $\efg$.
\end{corollary}
 
Finally, we derive an equivalence between the THPEs in EFGs and those in MAIMs. In order to do so, it suffices to prove an equivalence between perturbed versions of the corresponding games $\efg(\delta_k)$ and $\macid(\delta_k)$, which can easily be done via construction using \bfsf{efg2maim}, and then by applying Lemma \ref{lem:correspondence} and Corollary \ref{prop:BE}. 
 
\begin{proposition}
\label{prop:THPE}
    If $\efg\in\bfsf{maim2efg}(\macid)$ or $\macid = \bfsf{efg2maim}(\efg)$ then there is a natural mapping $f$ between $\efg$ and $\macid$ such that $\sigma$ is a THPE in $\efg$ if and only if any $\pi \in f(\sigma)$ is a THPE in $\macid$.
\end{proposition}

This series of equivalence results serves to justify MAIDs as an appropriate choice of game representation. Not only do they provide computational advantages over EFGs, we have shown that they preserve the most fundamental game-theoretic concepts commonly employed in EFGs.

\section{Implementation}
\label{sec:implementation}

K\&M showed that the explicit representation of dependencies between variables in MAIDs can substantially reduce the computational cost of finding an NE \cite{koller2003multi, milch2008ignorable}. In this section, we describe a modified version of their algorithm and use MAID subgames to find all pure SPEs (see also Appendix \ref{sec:codebase}). We show that MAID subgames exhibit the familiar subgame property of being useful for `generalised backwards induction' algorithms \cite{kaminski2019generalized}.

Beginning with an arbitrary policy profile $\pi(0)$ across all decision nodes in the original MAIM, $\macid$, we optimise decision rules associated to each $D \in \bm{D}$ by iterating backwards through a MAID subgame ordering from $\macid_m$ to $\macid_0$. In what follows, we write $\macid_i \prec \macid_j$ if $\macid_j$ is a proper MAID subgame of $\macid_i$, and $\bm{D}_k$ for the decision nodes in $\macid_k$. Several MAID subgames can have the same set of decisions, $\bm{D}_k$, so we choose a single MAID subgame $\macid_k$ (one with the fewest nodes $\bm{V}'$) for each $\bm{D}_k$ and discard the others. Each MAID in this ordering has an associated MAIM for each setting of the nodes which have a child in $\bm{V}'$. 

When considering a MAIM for $\macid_{m-i}$, the decision rules for all decision nodes in proper MAIM subgames of $\macid_{m-i}$ will have already been optimised and fixed in previous iterations, so these are now chance nodes in $\macid_{m-i}$. In addition, none of the decision nodes $\bm{D}_{m-i}$ in $\macid_{m-i}$ strategically rely on any of the decision nodes outside of $\macid_{m-i}$. Therefore, this step is localised to computing the optimal decision rules only for $\bm{D}_{m-i}$.

The next step depends on $\vert \bm{D}_{m-i} \vert$. If only one decision node $D \in \bm{D}^j$ remains, as in Figure \ref{fig:taxiMACID} c) for example, then its optimal decision rule is that which maximises player $j$'s expected utility in each of the MAIMs (for each value of $D^1$) for this MAID subgame. If $\vert \bm{D}_{m-i} \vert > 1$, the relevance graph of $\macid_{m-i}$ is cyclic and so the decision nodes strategically rely on one another. We must therefore call a subroutine: the MAIMs for the MAID subgame induced by the policy profile at that step, $\macid(\pi_{-\bm{D}_{m-i}}(i))$, are converted to EFGs to be solved using Gambit \cite{mckelvey2006gambit}. Algorithm \ref{algo:SPE} shows the full procedure.

\begin{algorithm}
    \caption{}
    \label{algo:SPE}
      \begin{algorithmic}[1]
        \Statex \textbf{Input:} MAIM $\macid = (\bm{N}, \bm{V}, \bm{E}, \theta)$
        \Statex \textbf{Output:} SPE $\pi$
        \State initialise $\pi(0)$ as an arbitrary fully mixed policy profile
        \State compute an ordering $\prec$ over the subgames $\macid_0,\dots,\macid_m$ in $\macid$
        \For{i = 0 to m-1}
            \State compute a best response policy profile $\pi^*_{\bm{D}_{m-i}}$ for all decision \hspace*{4mm} nodes in $\bm{D}_{m-i}$ using $\macid(\pi_{-\bm{D}_{m-i}}(i))$
            \State $\pi(i+1) \gets (\pi_{-\bm{D}_{m-i}}(i),\pi^*_{\bm{D}_{m-i}})$
        \EndFor
    \State \textbf{return} $\pi(m)$
    \end{algorithmic}
\end{algorithm}

It is more efficient to pass the EFGs in the algorithm's subroutine to an EFG solver such as Gambit \cite{mckelvey2006gambit}, rather than passing an EFG for the entire original MAIM. In the induced MAIM $\macid(\pi_{-\bm{D}_{m-i}}(i))$, all decision nodes in the proper MAIM subgames have been converted into chance nodes. In our MAID to EFG transformation we need only split on decision nodes and their informational parents, so the size of the EFG is exponential in $\vert \Fa_{\bm{D}_{m-i}} \vert$. As the time complexity of solving an EFG depends on its size, the cost of solving a MAIM using Algorithm \ref{algo:SPE} is never greater than solving an equivalent EFG representation of the original game, and is exponentially faster in many cases.

Our open-source Python codebase\footnote{Available online at \href{https://github.com/causalincentives/pycid}{\texttt{https://github.com/causalincentives/pycid}}.} implements this procedure, provides methods for finding and plotting MAIDs $\macid$, along with $Rel(\macid)$ and $ConRel(\macid)$, and converts any MAIM into an EFG to be used with Gambit. Our aim is to provide the necessary computational tools for researchers and practitioners to develop further applications of MAIDs.

\section{Discussion and conclusions}
\label{sec:discussion}
This work has extended previous results on MAIDs by introducing the concept of a MAID subgame and a range of key equilibrium refinements. K\&M argued that MAIDs offer several benefits \cite{koller2003multi}. First, MAIDs can represent games more concisely than EFGs. Second, because a parametrised MAID is a probabilistic graphical model, the probabilistic dependencies between chance and decision variables can be exploited in order to identify whether decision nodes strategically rely on one another; we used this to define MAID subgames and our resulting equilibrium refinements. Separately, MAIDs can lead to substantial savings in the computational cost of finding an SPE; in Section \ref{sec:implementation}, we have described a modified version of an algorithm of K\&M and implemented it in an open-source codebase.

These benefits of MAIDS, coupled with the theoretical and practical contributions of this paper, provide a rich basis for future work. One avenue for such work that we are already pursuing is to extend the analysis of incentives \cite{Everitt2021, carey2020incentives} to the multi-agent setting by interpreting the directed edges in MAIDs causally. One could then investigate which variables in the graph each agent has an incentive to observe or control, and which reasoning patterns are involved \cite{pfeffer2007reasoning}, given that all of the agents in the MAID are playing a certain equilibrium refinement.


\begin{acks}
The authors wish to thank Ryan Carey for invaluable feedback and assistance while completing this work, as well as Zac Kenton, Colin Rowat, and several anonymous reviewers for helpful comments. Hammond acknowledges the support of an EPSRC Doctoral Training Partnership studentship (Reference: 2218880). Fox acknowledges the support of the EPSRC Centre for Doctoral Training in Autonomous Intelligent Machines and Systems (Reference: EP/S024050/1).
\end{acks}

\balance



\bibliographystyle{ACM-Reference-Format} 
\bibliography{ms}



\renewcommand\thesection{\Alph{section}}
\setcounter{section}{0}
\setcounter{lemma}{0}
\setcounter{proposition}{0}
\setcounter{theorem}{0}
\setcounter{corollary}{0}

\onecolumn

\section{Transformations between Game Representations}
\label{sec:transformations}

\subsection{MAIM to EFG}
\label{sec:macid2efg}

In Section \ref{sec:MAIDtoEFG}, we mentioned that there are many ways of converting a MAIM into an EFG. In this section, we provide an encoding transforming a MAIM, $\macid = (\bm{N}, \bm{V}, \bm{E}, \theta)$, into an EFG, $\efg = (\bm{N}, T, \bm{P}, \bm{D}, \lambda, \bm{I}, U)$. One can decide on whether one wants to keep record of the structure of the original MAID, or if one wants to optimise for complexity by deciding which of $\macid$'s nodes we split on in $\efg$. If one wants to be able to preserve all of the MAID's structure (e.g. the dependencies between variables) then the set of splitting nodes in the resulting tree is $\bm{S} = \bm{X} \cup \bm{D} \subset \bm{V}$. If instead, however, one wants to minimise the complexity of calculating game equilibria, then one only needs to split on the MAID's decision nodes and informational parents, $\bm{S} = \bm{D} \cup \bigcup_{D\in\bm{D}} \Pa_{D}$ \cite{pearl2014probabilistic}. This is because the EFG's tree size will be exponential in the size of this set, $\vert \bm{S} \vert$. The following procedure, which we refer to as \bfsf{maim2efg}, is based on that of K\&M \cite{koller2003multi}. Given a MAIM $\macid = (\bm{N}, \bm{V}, \bm{E}, \theta)$ we define an equivalent EFG $\efg = (\bm{N}, T, \bm{P}, \bm{D}, \lambda, \bm{I}, U)$ as follows:
\begin{itemize}
    \item Choose a topological ordering $\prec \: = S_1,\dots,S_n$ over all nodes in $\bm{S}$ such that if $S_j$ is a descendent of $S_i$, then $S_i \prec S_j$.\footnote{This topological ordering will, in general, be non-unique.} Further, define $\bm{S}_i^\prec = \{S' : S' \prec S_i\}$ to denote the set of nodes in $\bm{S}$ that come before $S_i$ in $\prec$.
    
    \item The set of players, $\bm{N}$, in $\efg$ is the same as in $\mathcal{M}$.
    
    \item The tree $T$ is a symmetric tree with each path containing splits over all the variables in $\bm{S}$ in the order defined by $\prec$. 
    
    \item For all nodes $S \in \bm{S}$, if $S$ is a chance node in $\macid$, add it to $\efg$'s set of chance nodes $\bm{V}^0$. Else, if $S$ is a decision node belonging to some player $S \in \bm{D}^i$, add $S$ to player $i$'s set of nodes in $\efg$, $\bm{V}^i$.
    
    \item Label each node $Z$ in $T$ with an instantiation $\mu(Z)$ corresponding to the values taken by each EFG node (i.e., the branch followed from this node) on the path from the tree's root node $R$ to $Z$.
    
    \item For every node $Y$ in $T$ such that its corresponding node $S_Y \in \bm{S}$ is a chance node in $\mathcal{M}$, we determine the probability distribution $P_Y : \Ch_{Y} \rightarrow [0,1]$ over its children $Z \in \Ch_Y$ (each child corresponds to a value $s_y \in \dom(S_Y)$) by querying that node's CPD table with the instantiation label at $Y$. In other words, $P_Y(Z) \coloneqq \Pr(Y=s_y\mid\bm{S}_Y^\prec = \mu(Z))$. This equation can be simplified because in a Bayesian network the CPD for each node $S_Y$ only depends on the values of its parents, $\pa_{S_Y}$. Therefore, we have:
        $$P_Y(Z) \coloneqq \Pr(S_Y=s_y\mid\bm{S}_Y^\prec = \mu(Z))
         = \Pr(S_Y = s_y \mid \bm{S}_Y^\prec[{\Pa_{S_Y}}] = \mu(Z)[\pa_{S_Y}])$$
where $\bm{S}_Y^\prec[\Pa_{S_Y}]$, $\mu(Z)[\pa_{S_Y}]$ are the restriction of $\bm{S}_Y^\prec$ and $\mu(Z)$ to the variables $\Pa_{S_Y}$ and corresponding values $\pa_{S_Y}$. If $P_Y(Z) = 0$ for some child $Z$, then remove that child from the EFG.
    
    \item The set of available decisions $\bm{D}^i_j$ for player $i$ at node $V^i_j$ in the EFG is given by the domain of the corresponding decision node $D^i$ in the MAIM: $\bm{D}^i_j \coloneqq \dom(D^i)$.
    
    \item  Define $\lambda: \bm{V} \times \bm{V} \rightarrow \bm{D}$ for each $Y$ corresponding to a decision node in the MAIM ($Y \in \bm{V}^1 \cup \dots \cup \bm{V}^n$) and $Z \in \Ch_Y$ such that $\lambda(Y,Z) = \mu(Z)[s_Y]$ to label the outcome of the decision. 
    
    \item Define an equivalence relation $\sim$ over  $\bm{V}^1 \cup \dots \cup \bm{V}^n$ such that $Y \sim Z$ if and only if $S_Y = S_Z = S$ and $\mu(Y)[\pa_{S_Y}] = \mu(Z)[\pa_{S_Z}]$. Then the set of information sets $I^i$ for each player $i \in \bm{N}$ is simply the quotient set $\bm{V}^i/\sim$ -- the set of $\sim$ equivalence classes partitioning $\bm{V}^i$. In other words, two nodes $Y, Z \in \bm{V}$ which correspond to MAIM decision nodes are in the same information set $I^i_j$ for player $i$ if and only if their instantiation labels, restricted to their MAIM parents, are the same $\mu(Y)[\pa_{S_Y}] = \mu(Z)[\pa_{S_Z}]$.
    
    \item The payoff $U: \bm{L} \rightarrow \mathbb{R}^n$ for each leaf $L\in \bm{L}$ of the EFG's tree is $U(L) = (u^1,\dots,u^n)$, where $u^i \coloneqq \sum_{U^i_j \in \bm{U}^i} \mu(L)[u^i_j]$ and $u^i_j$ is the value for each utility node $U^i_j \in \bm{U}^i$ which is recorded as part of leaf's full instantiation label $\mu(L)$.
\end{itemize}

\subsection{EFG to MAIM}
\label{sec:EFG2MAID}

In this section, we detail the construction, which we refer to as \bfsf{efg2maim}, that underlies our equivalence results. Unlike in the case for \bfsf{maim2efg}, we show how every EFG can be converted to a \emph{unique}, canonical equivalent MAIM. In order to do this, it will be helpful to define the notion of an \emph{intervention set}. Note that in what follows we assume that the EFGs do not include instances of absentmindedness (where a path from the root of the tree to one of its leaves passes through an information set more than once); this restriction is lifted in our discussion of absentmindedness in Appendix \ref{sec:absent}.

\begin{definition}
\label{def:intervention_set}
    An \textbf{intervention set} $J$ in an EFG $\efg$ is either a set of chance nodes or a set of information sets belonging to the same player, such that:
    \begin{itemize}
        \item Each node in the set has the same number of children.
        \item No path from the root of $\efg$ to one of its leaves passes through $J$ more than once.
        \item For any two information sets $I$ and $I'$ in $J$, the same knowledge regarding the paths taken from the root of $\efg$ to $I$ and $I'$ must be available at both $I$ and $I'$.
    \end{itemize}
\end{definition}

The intuition behind this definition is that an intervention set represents a single variable in the real world. For example, while flipping two fair coins independently would result in an EFG with three nodes (excluding utilities), the two children of the root would represent a single event: the flipping of the second coin. In general, without further domain knowledge about the EFG, one cannot tell whether two separate chance nodes or information sets represent the same variable in the real world, and so by default every intervention set in an EFG is a singleton, containing only one chance node or one information set. However, if one does possess this knowledge (as is often the case), then it can allow one to form a more compact version of the equivalent MAIM (as shown by example in Section \ref{sec:EFGtoMAID}). Given an EFG $\efg = (\bm{N}, T, \bm{P}, \bm{D}, \lambda, \bm{I}, U)$ (including intervention sets, such that every non-utility node belongs to a single intervention set) we define an equivalent MAIM $\macid = (\bm{N}, \bm{V}, \bm{E}, \theta)$ as follows:

\begin{itemize}
    \item The set of players $\bm{N}$ remains the same in $\macid$ as in $\efg$.
    
    \item Initialise the MAID's graph $(\bm{V}, \bm{E})$ as $T$, where edges are directed from parents to children.
    \item For each of $\efg$'s chance nodes $V \in \bm{V}^0$, label each outgoing edge from $V$ with a value $v$. Every other node with an outgoing edge is labelled, by $\lambda$, with the decision corresponding to that edge. Thus, let $\dom(V)$ contain the labels of the outgoing edges from $V$.
    \item For each variable $V \in \bm{V}$ let $\rho_V$ be the unique path formed by the sequence of labels from the root $R$ of $\efg$ to the corresponding $\efg$ node for $V$ and let $\rho_V[V']$ denote the label of the outgoing edge from node $V'$ on the path $\rho_V$.
    \item For each information set $I^i_j$ in $\efg$, let $\rho^\prec(I^i_j)$ be the set of paths from $R$ into the nodes of $I^i_j$.
    Next, define a function $\mu : \bm{I} \rightarrow 2^{\bm{V}}$ that maps each information set $I^i_j$ to the set of variables in $\Anc_{I^i_j}$ whose outgoing label is the same in every path in $\rho^\prec(I^i_j)$. Note that by Definition \ref{def:intervention_set}, if information sets $I$ and $I'$ are in the same intervention set, then the nodes whose values are in $\mu(I)$ are the same as the nodes whose values are in $\mu(I')$.
    \item We then consider $\efg$'s chance, decision, and leaf nodes in turn:
    \begin{itemize}
        \item For every outgoing edge with some label $v_j$ from a chance node $V \in \bm{V}^0$, add the label $(v_j \mid \rho_{V_j}) : p$ where $p = P_j(v_j)$.
        \item For every information set $I^i_j$ and each variable $D \in I^i_j$, add the label $(d \mid \rho_D[\mu(I^i_j)] ) : \_$ where $\rho_D[\mu(I^i_j)]$ denotes the labels of $\rho_D$ restricted to those in $\{\rho_D[V] : V \in \mu(I^i_j)\}$ and $\_$ is a placeholder to be parametrised by a decision rule.
        \item For each leaf variable $L \in \bm{L}$ with payoff vector $U(L) = u$, split $L$ into $n$ utility nodes $U^1,\dots,U^n$ (duplicating incoming edges) with labels $(u[i] \mid \rho[L]) : 1$ respectively, where $u[i]$ is the $i^\text{th}$ entry in $u$.
    \end{itemize}
    \item Given these labellings we proceed by merging variables according to the intervention sets:
    \begin{itemize}
         \item Merge each information set $I^i_j$ into a single variable $D_j \in \bm{D}^i$, collecting the labels $(v \mid \rho_{V}) : p$ for each $V \in I^i_j$ and retaining all incoming and outgoing edges.
        \item Begin by adding a directed edge from every $V' \in \Anc_V$ to $V$ for every node $V$. Then for every variable $D_j$ corresponding to an information set $I^i_j$, remove all the incoming edges from variables that are not in $\mu(I^i_j)$.
        \item Merge each group of nodes, collecting their incoming and outgoing labelled edges, that belong to the same intervention set.
        \item Merge any utility nodes $U^i_j$ and $U^i_k$ belonging to the same player $i$ that have the same sets of incoming edges, and collect their labels.
    \end{itemize}
    \item We then let $\bm{V} \coloneqq \bm{X} \cup \bm{D} \cup \bm{U}$ where $\bm{X} = \bm{V}^0$, $\bm{D}$ is the union of variable sets $\bm{D}^i$ defined above, and $\bm{U}$ is the collection of utility variables, also defined above.
    \item $E$ is the set of edges in the graph defined above.
    \item We conclude by defining a set of CPDs, $\Pr(\bm{V} \mid \pa_V)$, for every chance and utility node $\bm{X} \cup \bm{U} \in \bm{V}$. 
    \begin{itemize}
        \item For certain instantiations, $\pa_V$, the CPD over $V$ is undefined because $\pa_V$ does not represent a path through the original game tree $T$. For non-decision variables, this is dealt with by simply adding a null value $\perp$ for every variable in $\bm{X}$ and a value of 0 for every variable in $\bm{U}$.
        \item Recall the labels of the form $(v \mid \rho_V) : p$ and let $lab(V)$ denote the set of $V$'s labels. 
        \begin{itemize}
            \item For each variable $V \in \bm{V} \setminus \bm{D}$, we define $\Pr(V = v\mid \Pa_V = \pa_V) \coloneqq \sum_{(v \mid \pa_V) : p ~\in~ lab(V)} p$.
            \item For any $\pa_V$ such that for all $v \in \dom(V)$, $\Pr(V = v\mid \Pa_V = \pa_V) = 0$, we set $\Pr(V = \perp\mid \Pa_V = \pa_V) = 1$ if $V \in \bm{X}$ and $\Pr(V = 0\mid \Pa_V = \pa_V) = 1$ if $V \in \bm{U}$.
        \end{itemize}
    \end{itemize}
    \item By construction, $\Pr(\bm{X},\bm{U} : \bm{D}) = \prod_{V \in \bm{V} \setminus \bm{D}} \Pr(V \mid \Pa_V)$ therefore forms a partial distribution over the nodes in $\macid$. 
    \item For decision variables, we now see that the only changes in parametrisations of $\pi(D \mid \pa_D)$ that can have any effect on any of the other variables are those that occur under settings $\pa_D$ such that there exists a strategy $\sigma$ and path $\rho$ in $\efg$ capturing all values in $\pa_D$ with $\Pr^\sigma(\rho) > 0$. In other words, the possible parametrisations of $\pi(D \mid \pa_D)$ correspond to parametrisations of $(d \mid \rho_D[\mu(I^i_j)] ) : \_$.
\end{itemize}

\subsection{Absentmindedness}
\label{sec:absent}
 An EFG $\efg$ with absentmindedness has at least one path from the root to some other vertex in $\efg$ that passes through the same information set more than once \cite{piccione1997interpretation}. In an EFG, a behavioural strategy defines a single distribution over decisions per information set; this corresponds to a single decision rule for a decision node in the MAIM. In a non-absentminded game, a behavioural strategy at an information set only corresponds to one \textit{decision instance} (as that information set is encountered only once for any path through the EFG). However, in $\efg$ the player may have to apply a behavioural strategy for a single information set at multiple decision instances.

In our usual transformation \bfsf{efg2maim}, a decision node $D$ in a MAID corresponds to an information set, and thus a single behavioural strategy. Therefore, there is no way for the decision choice to be different for each decision instance. To handle absentmindedness, we can replace $D$ with a set of nodes $D \cup \bm{X}^D$. $D$ is the decision node that represents the behavioural strategy with a decision rule, $\bm{X}^D$ is a set of chance nodes representing decision instances, and there is an edge from $D$ to all $X^D_j \in \bm{X}^D$. $D$'s outgoing edges in the original MAID now flow through some $X^D_j \in \bm{X}^D$, but incoming edges to $D$ remain the same. There is an additional edge $X^D_j \rightarrow X^D_k$ if and only if there is a path in the game tree from the root of $\efg$ that passes through the decision node represented by $X^D_j$ and then the decision node represented by $X^D_k$. Without absentmindedness, there is no reason to do this as $\vert \bm{X}^D \vert = 1$ and so the nodes $D \cup \bm{X}^D$ can be merged; however, with absentmindedness $\vert \bm{X}^D \vert > 1$ and so we must allow for the possibility of the player making different decision choices at different decision instances within the same information set. This can be seen, for example, in Figure \ref{fig:absentminded} where the player may use the same (stochastic) behavioural strategy at both nodes but end up continuing ($c$) at the first junction and exiting ($e$) at the second.

 \begin{figure}[h]
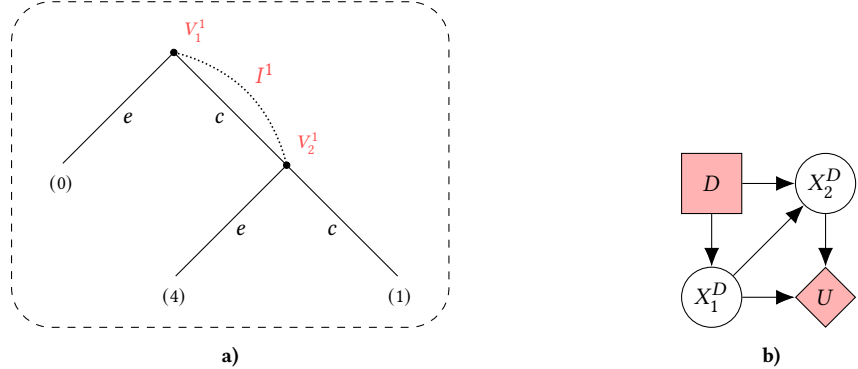

    \centering
    \begin{subfigure}[b]{0.5\linewidth}
        \vspace{0pt}
        \centering
        \begin{istgame}
            \xtdistance{15mm}{30mm}
            \istroot(0)(0,0)<45, red!70>{$V^1_1$}
            \istb{e}[br]{(0)}
            \istb{c}[bl]{}
            \endist
            \istroot(1)(0-2)<45, red!70>{$V^1_2$}
            \istb{e}[br]{(4)}
            \istb{c}[bl]{(1)}
            \endist
            \xtSubgameBox(0){(0)(0-1)(1-1)(1-2)(2-1)(2-2)}[black,inner sep = 18pt]
            \xtCInfoset(0)(1){\textcolor{red!70}{$I^1$}}[above right]
        \end{istgame}
        \caption*{a)}
    \end{subfigure}
    \begin{subfigure}[b]{0.3\linewidth}
        \vspace{0pt}
        \centering
        \begin{influence-diagram}
            \node (D) [decision, player1] {$D$};
            \node (XD2) [right = of D] {$X^D_2$};
            \node (XD1) [below = of D] {$X^D_1$};
            \node (U) [utility, below = of XD2, player1] {$U$};
            \edge {D} {XD1, XD2};
            \edge {XD1} {XD2};
            \edge {XD1, XD2} {U};
        \end{influence-diagram}
        \caption*{b)}
    \end{subfigure}
    \caption{An EFG a) and a MAID b) representing the classic `absentminded driver' game \cite{piccione1997interpretation}. $D$ represents the behavioural strategy for whether the driver should turn or continue at an exit. $X^D_1$ and $X^D_2$ are chance nodes representing the outcomes of the first and second decision instances (the two exits in the same EFG information set).}
    \label{fig:absentminded}
\end{figure}

\subsection{Subgames}
\label{sec:extra_subgames}

In Section \ref{sec:subgames}, we mentioned that a MAID may have more subgames than a corresponding EFG. We illustrate this here with an example in Figure \ref{fig:subgames2}, which shows that the number of MAID subgames can vastly exceed the corresponding EFG subgames.
In the EFG representation shown in Figure \ref{fig:subgames2} c), there are no proper subgames, because $D^2$ does not observe $X$.
However, from the MAID representation shown in Figure \ref{fig:subgames2} a), it is clear that $D^2$ does not need to know $X$ to choose an optimal policy.
Therefore, $D^2\to U^2$ can be factored out as a MAID subgame, shown top left in Figure \ref{fig:subgames2} b).
Similarly, the MAID representation reveals that we can solve the subgame corresponding to the $D^1$ subtrees independently, shown bottom left in Figure \ref{fig:subgames2} b).
More surprisingly, perhaps, MAID subgames need not be closed under descendants. Indeed, in the example in Figure \ref{fig:subgames2}, it is possible to solve for $D^1$ without knowing the policy for $D^2$.
This is represented by the MAID subgames in the right column of Figure \ref{fig:subgames2} b).
Additional subgames may of course also occur in much more complex games, in which the computational advantages of being able to flexibly factor the game into subgames are more keenly felt.

\begin{figure}[h]
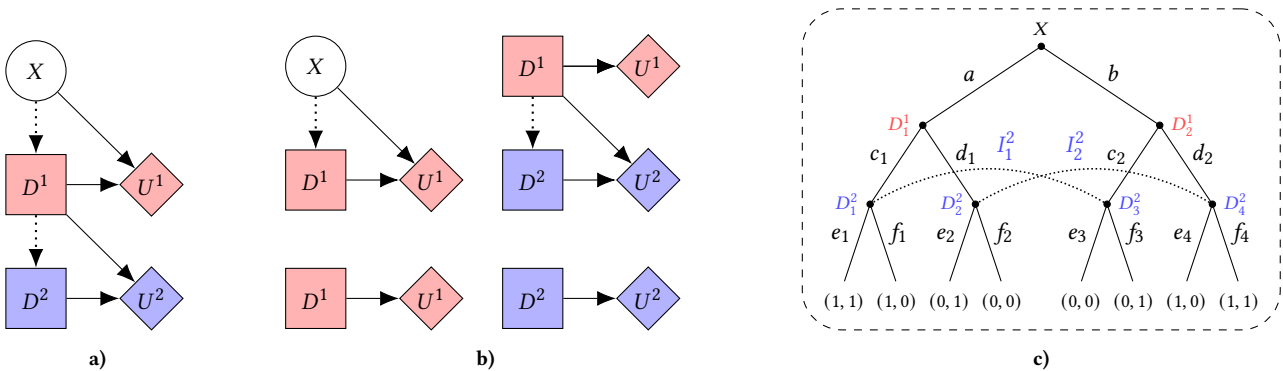

    \begin{subfigure}[b]{0.15\linewidth}
        \vspace{0pt}
        \centering
        \begin{influence-diagram}
        \node (X) {$X$};
        \node (D1) [below = of X, decision, player1] {$D^1$};
        \node (D2) [decision, player2, below = of D1] {$D^2$};
        \node (U1) [utility, player1, right = of D1] {$U^1$};
        \node (U2) [utility, player2, right = of D2] {$U^2$};
        \edge[information] {X} {D1};
        \edge {X, D1} {U1};
        \edge[information] {D1} {D2};
        \edge {D1, D2} {U2};
        \end{influence-diagram}
        \caption*{a)}
    \end{subfigure}
    \hspace{1cm}
    \begin{subfigure}[b]{0.3\linewidth}
        \vspace{0pt}
        \centering
        \begin{influence-diagram}
        \node (X) {$X$};
        \node (D1) [below = of X, decision, player1] {$D^1$};
        \node (U1) [utility, player1, right = of D1] {$U^1$};
        \edge[information] {X} {D1};
        \edge {X, D1} {U1};
        \node [phantom, below = of D1] {}; 
        \end{influence-diagram}
        \hspace{2mm}
        \begin{influence-diagram}
        \node (D1) [decision, player1] {$D^1$};
        \node (D2) [decision, player2, below = of D1] {$D^2$};
        \node (U1) [utility, player1, right = of D1] {$U^1$};
        \node (U2) [utility, player2, right = of D2] {$U^2$};
        \edge {X, D1} {U1};
        \edge[information] {D1} {D2};
        \edge {D1, D2} {U2};
        \node [phantom, below = of D2] {};
        \end{influence-diagram}
        
        \begin{influence-diagram}
        \node (D1) [decision, player1] {$D^1$};
        \node (U1) [utility, player1, right = of D1] {$U^1$};
        \edge {D1} {U1};
        \end{influence-diagram}
        \hspace{2mm}
        \begin{influence-diagram}
        \node (D2) [decision, player2] {$D^2$};
        \node (U2) [utility, player2, right = of D2] {$U^2$};
        \edge {D2} {U2};
        \end{influence-diagram}
        \caption*{b)}
    \end{subfigure}
    \hspace{1cm}
    \begin{subfigure}[b]{0.4\linewidth}
        \vspace{0pt}
        \centering
        \begin{istgame}[scale=0.7]
        \xtdistance{15mm}{45mm}
        \istroot(0){$X$}
        \istb{a}[al]
        \istb{b}[ar] 
        \endist
        \xtdistance{15mm}{20mm}
        \istroot(1)(0-1)<180, red!70>{$D^1_1$}
        \istb{c_1}[al]
        \istb{d_1}[ar] 
        \endist
        \istroot(2)(0-2)<0, red!70>{$D^1_2$}
        \istb{c_2}[al]
        \istb{d_2}[ar] 
        \endist
        \xtdistance{15mm}{10mm}
        \istroot(3)(1-1)<180, blue!70>{$D^2_1$}
        \istb{e_1}[al]{(1,1)}
        \istb{f_1}[ar]{(1,0)} 
        \endist
        \istroot(4)(1-2)<180, blue!70>{$D^2_2$}
        \istb{e_2}[al]{(0,1)}
        \istb{f_2}[ar]{(0,0)} 
        \endist
        \istroot(5)(2-1)<0,blue!70>{$D^2_3$}
        \istb{e_3}[al]{(0,0)}
        \istb{f_3}[ar]{(0,1)} 
        \endist
        \istroot(6)(2-2)<0,blue!70>{$D^2_4$}
        \istb{e_4}[al]{(1,0)}
        \istb{f_4}[ar]{(1,1)} 
        \endist
        \xtCInfoset(3)(5){\textcolor{blue!70}{$I^2_1$}}[above right]
        \xtCInfoset(4)(6){\textcolor{blue!70}{$I^2_2$}}[above left]
        \xtSubgameBox(0){(0)(1)(3-1)(3-2)(4-1)(4-2)(2)(5-1)(5-2)(6-1)(6-2)}[black,inner sep = 15pt, xshift=0pt, yshift=-2pt]
        \end{istgame}
        \caption*{c)}
    \end{subfigure}
    \caption{A MAID a) with four proper subgames b), none of which can be recognised in the EFG c). Note that each MAID subgame will have multiple instantiations (i.e.\ multiple MAIM subgames) depending on the values of the removed nodes.
    }
    \label{fig:subgames2}
\end{figure}

\section{Proofs}
\label{sec:proofs}

\begin{lemma}
    If $\efg\in\bfsf{maim2efg}(\macid)$ or
    $\macid =\bfsf{efg2maim}(\efg)$ then $\efg$ and $\macid$ are equivalent.
\end{lemma}
\begin{proof}

    In what follows we make the trivial assumption that any non-leaf node $V$ in $\efg$ has more than one child (otherwise we could simply remove such nodes), and that if $V$ is a chance node then all of its children are reached with positive probability (otherwise we could delete the sub-trees rooted at such children). The proof follows directly by construction using our procedures \bfsf{maim2efg} and \bfsf{efg2maim} respectively. 
    
    To see this, first suppose we have a MAIM, $\macid$, and $\efg$ is an EFG resulting from $\text{\bfsf{maim2efg}}(\macid)$. A decision rule $\pi_D$ defines a probability distribution over $\dom(D)$ conditional on some feasible decision context, $\pa_D$. Following the \bfsf{maim2efg} procedure, each feasible decision context is an instantiation of a set of variables which corresponds to one information set $I_j^i$ where $S_Y = D$ for all $Y \in I_j^i$, and for each $d \in \dom(D)$ there exists precisely one node $Z \in \Ch_Y$ such that $\lambda(Y,Z) = d$. A policy $\pi^i$ for player $i$ is simply the set of decision rules for all $D \in \bm{D}^i$. Thus, by construction, there is a one-to-one correspondence between each behavioural strategy $\sigma^i$ in $\efg$ and policy $\pi^i$ in $\macid$, where $\sigma^i_j(d) \coloneqq \pi_D(d\mid\pa_D)$. By reasoning analogously about the chance and utility nodes, it follows from this and our construction above that the expected utility for each player $i$ under a policy profile $\pi$ in $\macid$ is the same as the expected utility in $\efg$ under $\sigma$, $\mathcal{U}^i_\macid(\pi) = \mathcal{U}^i_\efg(\sigma)$.

    In the second direction, note first that the deterministic nature of the procedure guarantees uniqueness of the resulting MAIM. Then, suppose we have an EFG $\efg$ and $\macid$ is the MAIM that results from $\text{\bfsf{efg2maim}}(\efg)$. In our construction above, the incoming edges for each decision variable $D_j \in \bm{D}^i$ are precisely those that originate in the variables whose values player $i$ can determine when in the corresponding information set(s) $I^i_j$. Hence the strategy $\sigma^i_j$, which assigns a probability distribution over the set of available decisions $\bm{D}^i_j$ at node $V^i_j$, is determined as a function of $\pa_D$, where $\pa_D$ corresponds to a particular information set $I^i_j$ from which $D_j$ in the MAIM was created. Thus, given a strategy $\sigma^i$ in $\efg$ we define the \emph{corresponding} policy $\pi^i$ in $\macid$ by $\pi^i(d\mid\pa_D) \coloneqq \sigma^i_j(d)$. Note that unlike in the case for \bfsf{maim2efg}, there may be more possible policy profiles in $\macid$ compared to strategy profiles in $\efg$ because of the possibility of decision contexts (involving null values $\perp$) corresponding to impossible paths through the game tree. In other words, via our construction \bfsf{edf2maim} these decision contexts are \emph{null}. Hence, although the policy space may technically be larger, a player's decision rule under such contexts has no impact on their expected utility. Once again, therefore, we have that $\mathcal{U}^i_\macid(\pi) = \mathcal{U}^i_\efg(\sigma)$ for any policy $\pi$ corresponding to a strategy $\sigma$.
\end{proof}

Before giving proofs of the remaining results, we provide the relevant definitions of the corresponding equilibrium refinements in EFGs.

\begin{definition}
\label{def:EFGNE}
    A strategy profile $\sigma$ is a \textbf{Nash equilibrium (NE)} in an EFG $\efg$ if, for every player $i \in \bm{N}$, $\mathcal{U}^i_\efg(\sigma^{-i}, \sigma^i) \geq \mathcal{U}^i_\efg(\sigma^{-i}, \hat{\sigma}^i)$ for all $\hat{\sigma}^i \in \Sigma^i$.
\end{definition}

\begin{definition}
    A strategy profile $\sigma$ is a \textbf{subgame perfect equilibrium (SPE)} in an EFG $\efg$ if $\sigma$ is a NE in every EFG subgame of $\efg$, when restricted to the decision nodes in each EFG subgame.
\end{definition}

\begin{definition}
    A behavioural strategy profile $\sigma$ is a \textbf{trembling hand perfect equilibrium (THPE)} in an EFG $\efg$ if there is a sequence of perturbation vectors $\{\delta_k\}_{k\in\mathbb{N}}$ such that $\lim_{k \rightarrow \infty}\vert\delta_k\vert_\infty = 0$ and for each perturbed game $\efg(\delta_k)$ there exists an NE $\sigma_k$ such that $\lim_{k \rightarrow \infty} \sigma_k = \sigma$.  In $\efg$, a perturbation vector $\delta_k$ is a sequence of perturbations $\epsilon^d_j > 0$ with $\sum_{d \in D^i_j} \epsilon^d_j \leq 1$ for every information set $I^i_j$ such that in the perturbed game $\efg(\delta_k)$, each decision $d$ available at $I^i_j$ is played with probability at least $\epsilon^d_j$ (i.e., each strategy profile $\sigma_k$ is fully mixed).
\end{definition}

\begin{corollary}
    If $\efg\in\bfsf{maim2efg}(\macid)$ or $\macid = \bfsf{efg2maim}(\efg)$ then then there is a natural mapping $f$ between $\efg$ and $\macid$ such that $\sigma$ is an NE in $\efg$ if and only if any $\pi \in f(\sigma)$ is an NE in $\macid$.
\end{corollary}

\begin{proof}
    By Lemma \ref{lem:correspondence} we have that $\efg$ and $\macid$ are equivalent. Thus, let $f : \Sigma \rightarrow \Pi / \sim$ be a natural mapping  between $\efg$ and $\macid$ as defined in Definition \ref{def:equiv}.
    Next note that for any $\sigma$ and any $\pi,\pi' \in f(\sigma)$ we have $\mathcal{U}^i_{\macid}(\pi) = \mathcal{U}^i_{\macid}(\pi')$ for every player $i$. Finally, observe that by the equivalence between $\efg$ and $\macid$ we have that for any $\sigma^{-i}$ then $\sigma^i \in \argmax_{\hat{\sigma}^i} \mathcal{U}^i_{\efg}(\sigma^{-i},\hat{\sigma}^i)$ if and only if $\pi^i \in \argmax_{\hat{\pi}^i} \mathcal{U}^i_{\macid}(\pi^{-i},\hat{\pi}^i)$ where $\pi \in f(\sigma)$, and hence that $\sigma$ is an NE if and only if every $\pi \in f(\sigma)$ is an NE.
\end{proof}

\begin{proposition}
    If $\efg\in\bfsf{maim2efg}(\macid)$ or $\macid = \bfsf{efg2maim}(\efg)$ then there is a natural mapping $f$ between $\efg$ and $\macid$ such that, for every EFG subgame $\efg'$ in $\efg$ there is a MAIM subgame $\macid'$ in $\macid$ that is equivalent to $\efg'$ under the natural mapping $f$ restricted to the strategies of $\efg'$.
\end{proposition}

\begin{proof}
    We begin by proving existence. Recall that a subgame $\efg'$ in an EFG $\efg$ is a subtree that is closed under information sets and descendants. Let $\bm{V}'$ be the set of intervention sets overlapping with or contained in $\efg'$, which are therefore a subset of the variables $\bm{V}$ in the equivalent MAIM $\macid$. We first show that $\bm{V}'$ forms a subgame base, i.e. that $\bm{V}'$ is closed under r-reachability in $\macid$, and for any $X,Y\in \bm{V}'$ and any directed path $X \to \dots \to Y$ in $\macid$, all nodes on the path are also in $\bm{V'}$. Beginning with the first condition, it suffices to show that there exists no variable $Y \in \bm{Y} = \bm{V} \setminus \bm{V}'$ such that $Y$ is r-reachable from some node $D' \in \bm{D'} = \bm{D} \cap \bm{V}'$: recall that this means we would have $\hat{Y} \not\perp \bm{U}^i \cap \Desc_{D'} \mid \Fa_{D'}$, where $\hat{Y}$ is an extra parent node added to $Y$ in $\macid$ and we have $D' \in \bm{D}^i$ for some player $i$. Any path supporting such a dependency must have one of the following two forms:
    \begin{itemize}
        \item $\hat{Y} \rightarrow Y \rightarrow \cdots~ \bm{U}^i \cap \Desc_{D'}$. In this case, as $Y \notin \bm{V'}$     then, by either construction \bfsf{efg2maim} or \bfsf{maim2efg}, any node in $\efg$ corresponding to $Y$ must lie outside $\efg'$.
        Further, as $\efg'$ is closed under information sets, then the value of $Y$ must be observed by any decision node in $\efg'$ corresponding to $D'$, and thus there is an information link $Y \rightarrow D'$ in $\macid$ and so $Y \in \Fa_{D'}$. Hence, by conditioning on $\Fa_{D'}$ we block the path above, meaning there is no dependency.
        \item $\hat{Y} \rightarrow Y \leftarrow \cdots~ \bm{U}^i \cap \Desc_{D'}$. In this case, let $X$ be the first variable in the path (from left to right) that is a fork node, i.e., we have $\cdots \leftarrow X \rightarrow \cdots$. Such a variable must exist because we assume that utility variables do not have children. 
        As before, notice that by either construction \bfsf{efg2maim} or \bfsf{maim2efg} we have that $Y \notin \bm{V'}$ must lie outside $\efg'$, and as there is a directed path $X \to \dots \to Y$ in $\macid$ then $X$ must also lie outside $\efg'$. But then the fact that $\efg'$ is closed under information sets means that $X$ is observed by $D'$ and so conditioning on $\Fa_{D'}$ blocks the path above, meaning there is no dependency.
    \end{itemize}
    We next consider the second condition. Suppose that $X,Y\in \bm{V}'$ and there is a directed path $X \to \cdots \to Z \to \cdots \to Y$ in $\macid$. Then if $\efg \in \bfsf{maim2efg}(\macid)$ or $\macid = \bfsf{efg2maim}(\efg)$, any topological ordering $\prec$ over the nodes in $\macid$ must have $X \prec Z \prec Y$ and hence if there are nodes in $\efg'$ corresponding to both $X$ and $Y$ then as $\efg'$ is closed under descendants we must have a node corresponding to $Z$ in $\efg'$. This means that the intervention set containing $Z$ overlaps with $\efg'$ and so $Z \in \bm{V}'$ by the definition of $\bm{V}'$. 
    
    For the second part of the existence proof, we show that there is a setting $\bm{y}$ of $\bm{Y}$ which, when combined with $\bm{V'}$, leads to a subgame $\macid'$ of $\macid$ that is equivalent to $\efg'$. First, note that any node passed through on the path $\rho_{R'}$ from the root $R$ of $\efg$ to the root $R'$ of $\efg'$ must correspond to a variable in $\bm{Y}$. Let $\bm{y}$ be a setting of $\bm{Y}$ that is consistent with the path from $R$ to $R'$, and let $\macid'$ be the result MAIM subgame that is obtained by combining $\bm{V'}$ with $\bm{y}$. Observe that by the argument argument above, the decision variables in $\bm{D'} \subseteq \bm{V'}$ are precisely those corresponding to the information sets in $\efg'$, and moreover that any decision context of any $D' \in \bm{D'}$ that does not correspond to some information set in $\efg'$ is, by definition, not feasible and hence null. Thus, the partitioned policy space in $\macid'$ (where two policies belong to the same partition if and only if they differ only on those decision contexts that are null) is in one-to-one correspondence with the strategies in $\efg'$, just as was shown in the proof of Lemma \ref{lem:correspondence}. Further, for any equivalent strategy $\sigma$ in $\efg$ and $\pi$ in $\macid$, it follows directly by construction using \bfsf{maim2efg} or \bfsf{efg2maim} that:
    \begin{align*}
        \mathcal{U}^i_{\efg'}(\sigma) 
        &= \sum_\rho {P'}^\sigma(\rho) U(\rho[\bm{L}])[i]\\
        &= \sum_\rho {P}^\sigma(\rho \mid \rho_{R'}) U(\rho[\bm{L}])[i]\\
        &= \sum_{U_j \in \bm{U}^i}\sum_{u_j \in \dom(U_j)} \!\! u_j {\Pr}^\pi(U_j = u_j \mid \bm{y})\\
        &= \sum_{U_j \in \bm{U}^i}\sum_{u_j \in \dom(U_j)} \!\! u_j {\Pr'}^\pi(U_j = u_j)\\
        &= \mathcal{U}^i_{\macid'}(\pi)
    \end{align*}
    for every player $i$, where ${P}^\sigma(\rho \mid \rho_{R'})$ is the distribution over paths $\rho$ in $\efg$ conditional on $\rho$ containing $\rho_{R'}$ as a sub-path. Hence, not only do we have a one-to-one correspondence $f$ between strategies in $\efg'$ and (a partition over) the policies in $\macid'$, we see that for any $\pi \in f(\sigma)$ we have that $\mathcal{U}^i_{\efg'}(\sigma) = \mathcal{U}^i_{\macid'}(\pi)$ for all $i \in \bm{N'} = \{i \in \bm{N} \mid \bm{D}^i \cap \bm{V'} \neq \varnothing \}$ and hence that $\efg'$ and $\macid'$ are equivalent, and thus that $f$ forms a natural mapping between $\efg'$ and $\macid'$.
\end{proof}

\begin{corollary}
   If $\efg\in\bfsf{maim2efg}(\macid)$ or $\macid = \bfsf{efg2maim}(\efg)$ then there is a natural mapping $f$ between $\efg$ and $\macid$ such that if any $\pi \in f(\sigma)$ is an SPE in $\macid$, then $\sigma$ is an SPE in $\efg$.
\end{corollary}

\begin{proof}
     The corollary can be seen to follow immediately from combining the definitions of SPEs in EFGs and MAIMs with Proposition \ref{prop:subgames} and Corollary \ref{prop:BE}. Concretely, if $\efg \in \bfsf{maim2efg}(\macid)$ or $\macid = \bfsf{efg2maim}(\efg)$, and if any $\pi \in f(\sigma)$ is an SPE in $\macid$ (for some natural mapping $f$ between $\efg$ and $\macid$ satisfying Proposition \ref{prop:subgames}) then $\pi$ is an NE in any subgame of $\macid$, and as every subgame of $\efg$ is equivalent to a subgame of $\macid$ under $f$, then we must have that $\sigma$ is an NE in each subgame of $\efg$ (by Corollary \ref{prop:BE}).
\end{proof}

\begin{proposition}
    If $\efg\in\bfsf{maim2efg}(\macid)$ or $\macid = \bfsf{efg2maim}(\efg)$ then there is a natural mapping $f$ between $\efg$ and $\macid$ such that $\sigma$ is a THPE in $\efg$ if and only if any $\pi \in f(\sigma)$ is a THPE in $\macid$.
\end{proposition}

\begin{proof}
    Given Lemma \ref{lem:correspondence} and Corollary \ref{prop:BE}, it suffices to show that the perturbed MAIM resulting from $\text{\bfsf{efg2maim}}(\efg(\delta_k))$ is equivalent to $\efg(\delta_k)$, and similarly that any perturbed EFG resulting from $\text{\bfsf{maim2efg}}(\macid(\delta'_k))$ is equivalent to $\macid(\delta'_k)$. It can readily be seen from the constructions detailed in Appendices \ref{sec:macid2efg} and \ref{sec:EFG2MAID} respectively that the entry $\epsilon^d_j$ of the perturbation vector $\delta_k$ for $\efg$ corresponds precisely to the entry $\epsilon^d_{\pa_D}$ in the perturbation vector $\delta'_k$ for $\macid$. The only extra entries that $\delta'_k$ may have over $\delta_k$ are those correspond to null decision contexts $\pa_D$, which are are irrelevant for the determination of NEs in the perturbed MAIM.
\end{proof}

\section{Codebase}
\label{sec:codebase}

In this section, we demonstrate the usage of our Python codebase, available at \href{https://github.com/causalincentives/pycid}{\texttt{https://github.com/causalincentives/pycid}}. We provide a number of screenshots of a Jupyter notebook showing how to instantiate MAIDs and MAIMs, plot attributes, and how to find all of the MAIM's SPEs. We first give an overview of some of our codebase's key methods (Appendix \ref{sec:keymethods}) before demonstrating how to instantiate a particular MAID or MAIM (Appendix \ref{sec:insMACID}). We then show how we compute all pure policy SPEs (Appendix \ref{sec:NEcomp}) and explain the usefulness of interfacing with Gambit \cite{mckelvey2006gambit} (Appendix \ref{sec:gambit}).

\subsection{Key Methods} 
\label{sec:keymethods}

The codebase relevant for what is discussed in this paper is centred around a \texttt{MACID} class. A MACID is a multi-agent causal influence diagram, which is a slightly different model than a MAID because the edges represent every causal relationship between the random variables chosen to be endogenous variables in the model. However, because MACIDs subsume MAIDs (in the sense of Pearl's `causal hierarchy' \cite{pearl2009causality}), for the purpose of this discussion, one can interpret this object to be the same as that defined in Definition \ref{def:MAID}. Our \texttt{MACID} class inherits from pgmpy's \texttt{BayesianModel} class \citep{ankan2015pgmpy} because when a policy profile $\pi$ has been specified in a MAID $\macid$, the induced MAID $\macid(\pi)$ is a Bayesian network with joint probability distribution $\Pr^\pi$. Some of the most important methods belonging to this class, and relating to the work in this paper, perform the following operations:

\begin{itemize}
    \item Plot the MAID using NetworkX \citep{hagberg2008exploring}.
    \item Find and plot the MAID's relevance graph (Definition \ref{def:relgraph}) with its maximal SCCs (as in, for example, Figure \ref{fig:road}), and its condensed relevance graph.
    \item Find a MAIM's pure SPEs using Algorithm \ref{algo:SPE}.
    \item Find all MAID and MAIM subgames.
    \item Convert any MAID into an EFG using the procedure in Appendix \ref{sec:macid2efg}, and into a format that can be used with Gambit \citep{mckelvey2006gambit}.
\end{itemize}

Making use of the existing literature on MAIDs and CIDs and to advantage researchers for future work on MAIDs (and MACIDs), we have also implemented methods to:

\begin{itemize}
    \item Determine if various graphical criteria are satisfied such as:
    \begin{itemize}
        \item Whether a node faces incentives in single-agent CIDs \cite{Everitt2021, carey2020incentives}.
        \item Whether a decision node in a MAID faces one of Pfeffer and Gal's reasoning patterns \cite{pfeffer2007reasoning}.
    \end{itemize}
    \item Prune the original MAID's edges based on these reasoning patterns and K\&M's concept of ignorable information \cite{milch2008ignorable}.
\end{itemize}

\vspace{0.75cm}
\begin{minipage}[]{0.6\textwidth}
    \centering
    \includegraphics[scale=0.45]{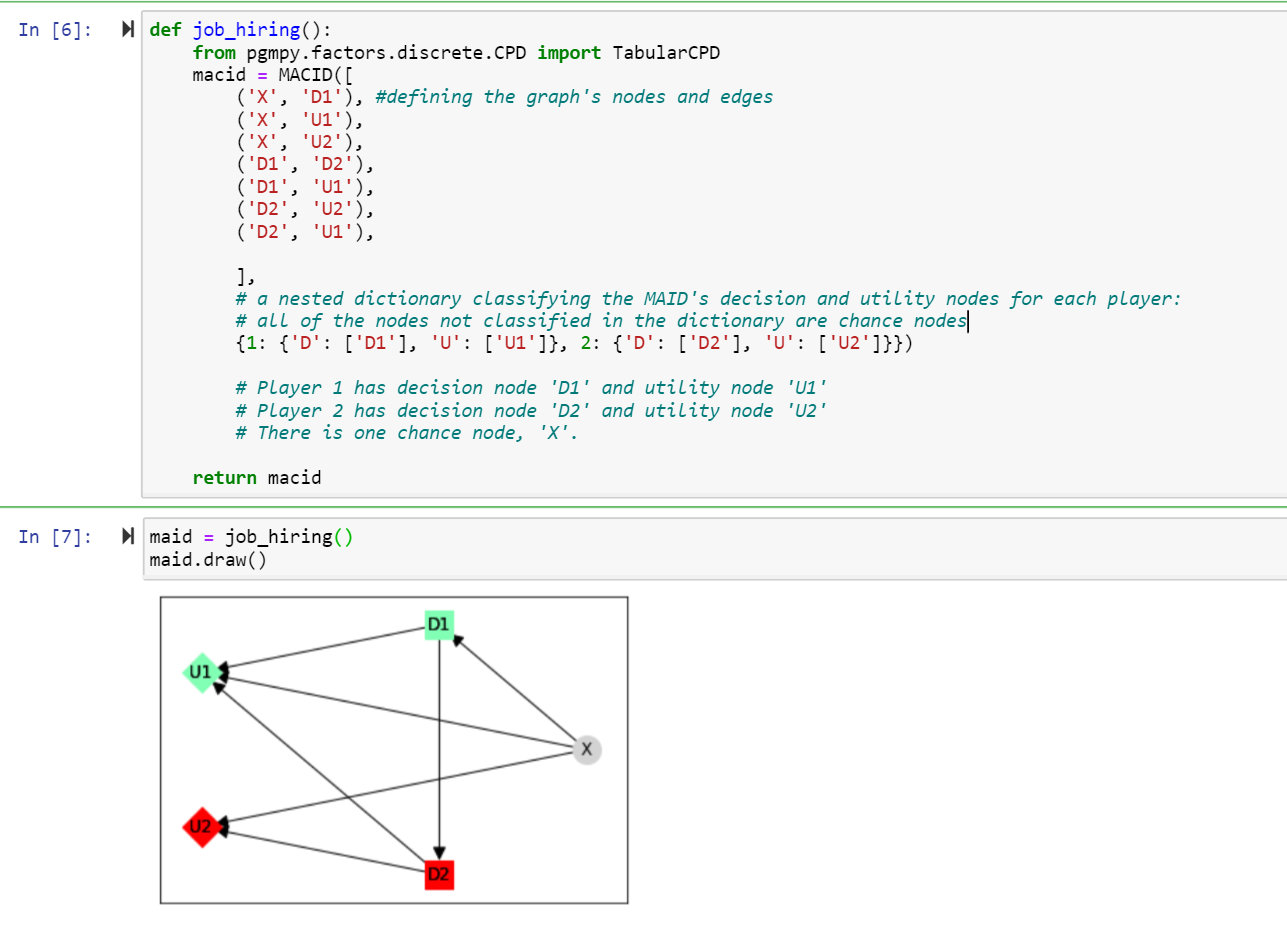}
    \captionof{figure}{Instantiating and plotting a \texttt{job\_hiring} object for Example \ref{ex:hiring}'s MAID as an instance of our MAID class.}
    \label{fig:jobhire}
\end{minipage}
\hspace{0.05\textwidth}
\begin{minipage}[]{0.3\textwidth}
    \centering
    \begin{minipage}{\textwidth}
    \centering
        \includegraphics[scale=0.4]{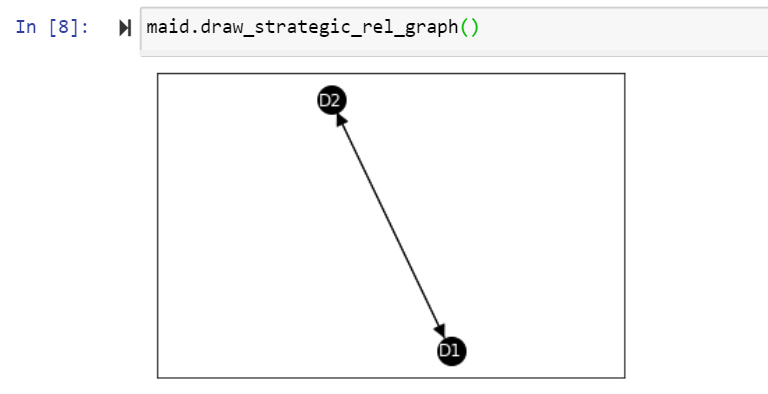}
        \captionof{figure}{The cyclic relevance graph for Example \ref{ex:hiring}'s MAID.}
        \label{fig:jh_rel}
    \end{minipage}
    \vspace{0.15cm}
    
    \begin{minipage}{\textwidth}
        \centering
        \includegraphics[scale=0.3]{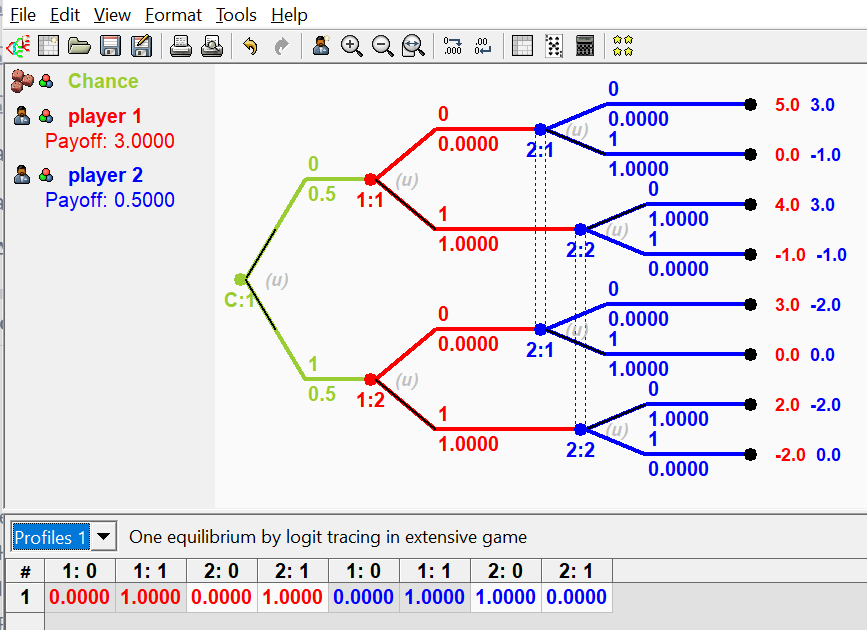}
        \captionof{figure}{Gambit's GUI EFG representation of Example \ref{ex:hiring}.}
        \label{fig:gambit}
    \end{minipage}
\end{minipage}
\vspace{0.75cm}

\vspace{0.5cm}
\begin{minipage}[]{0.45\textwidth}
    \centering
    \includegraphics[scale=0.45]{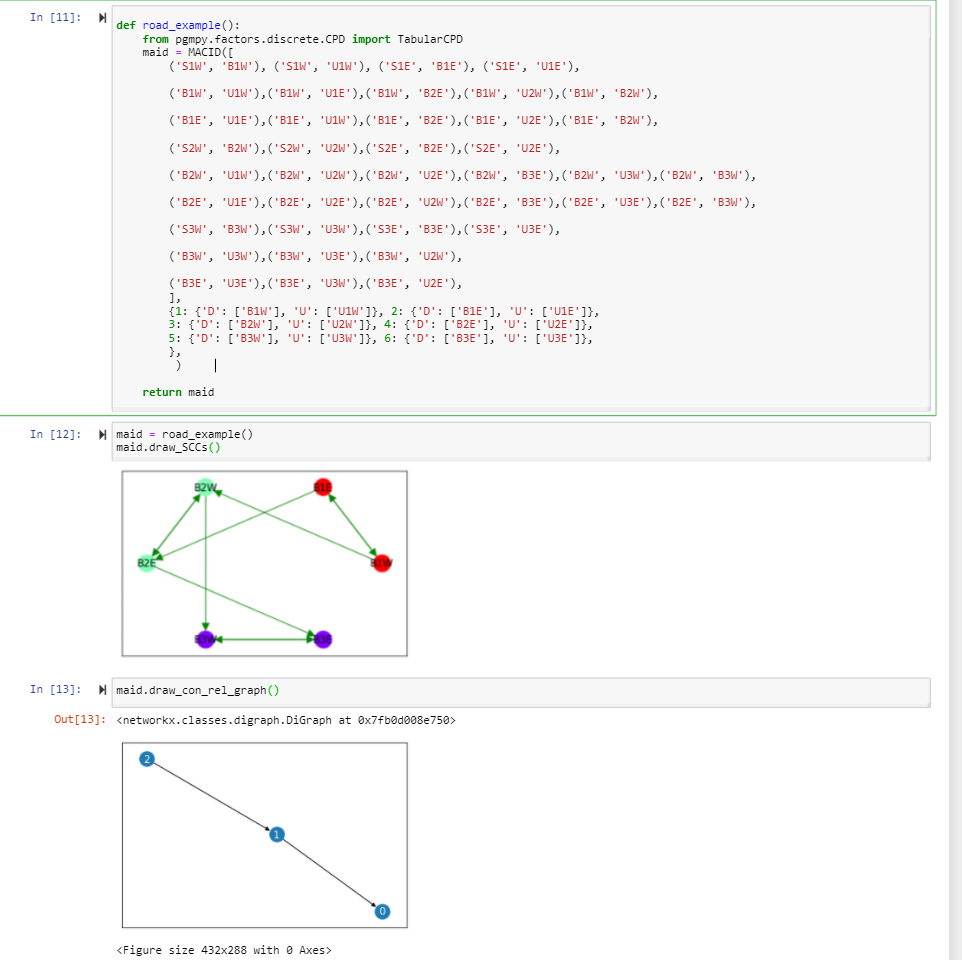}
    \captionof{figure}{Instantiating, plotting the relevance graph (with each maximal SCC coloured uniquely), and plotting the condensed relevance graph for K\&M's `Road Example' \cite{koller2003multi}.}
\label{fig:road}
\end{minipage}
\hspace{0.05\textwidth}
\begin{minipage}[]{0.45\textwidth}
    \centering
    \includegraphics[scale=0.45]{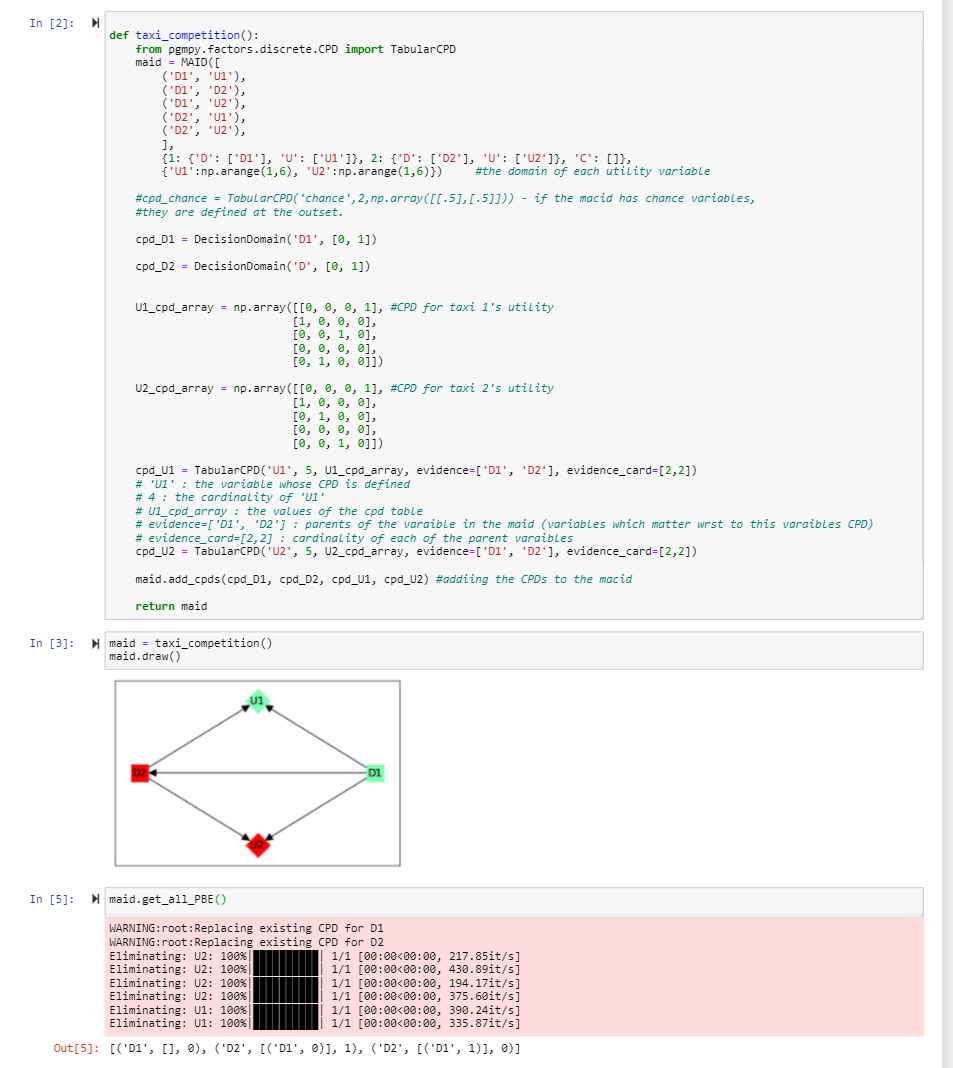}
    \captionof{figure}{Instantiating a MAIM for Example \ref{ex:taxis}. All SPEs are returned as nested lists. Each list contains triples to return the policy CPDs: the first argument gives the decision node, the second argument is a list giving the feasible decision context that is being conditioned on, the third gives the decision node's action. Here, action $0$ corresponds to selecting the expensive hotel and action $1$ is the cheap hotel.}
    \label{fig:taxi}
\end{minipage}
\vspace{0.5cm}

\subsection{Instantiating a MAID} \label{sec:insMACID}

Recall that a MAID consists of a graph ($\bm{V}, \bm{E}$) to represent the conditional dependencies between random variables and a set of players $\bm{N}$. A MAIM adds a parametrisation, $\theta \in \Theta$, to this MAID to specify with CPDs the precise relationship between variables in the MAID. Therefore, if we are only interested in MAIM's structure, we need only specify the graph's nodes (including their partition into chance, decision, and utility variables, and their allocation to respective players) and edges. Figure \ref{fig:jobhire} shows how we define, instantiate, and then plot the MAID for Example \ref{ex:hiring}. Decision and utility nodes are square- and diamond-shaped respectively. The nodes belonging to each agent are coloured uniquely and chance nodes are grey. Because r-reachability (Definition \ref{def:rreachable}) is a graphical criterion, we can also plot the MAID's relevance graph before specifying the MAID's parametrisation. 

To demonstrate a slightly more complex MAID, we use K\&M's `Road Example' \cite{koller2003multi}. As the relevance graph for this MAID contains cycles, it is helpful to have methods that visualise the relevance graph's distinct maximal SCCs and the MAID's condensed relevance graph (Figure \ref{fig:road}). This shows clearly that although the initial game is highly complex, the MAID formalism allows us to exploit the dependencies between the decision variables because there are 3 maximal SCCs. The MAID's condensed relevance graph shows that $(B3W, B3E)$ are the decision nodes in the smallest proper MAID subgame, and then $(B2W, B2E)$, before finally $(B1W, B1E)$ are only decision nodes in the largest subgame, which is identical to the original MAID. This is also the ordering that Algorithm \ref{algo:SPE} would follow once the MAID is parameterised as a MAIM. 

If we want the MAID to be parametrised as a MAIM (such as for interfacing with Gambit \cite{mckelvey2006gambit} or computing SPEs), we must also define:
\begin{itemize}
  \item Each decision node's domain.
  \item Domains and CPDs for each chance and utility node. Definition \ref{def:MAIM} states that these CPDs must reflect the fact that the value of a utility node is a deterministic function of the value of its parents.
\end{itemize}

\subsection{Finding Subgame Perfect Equilibria} \label{sec:NEcomp}

We will explain how Algorithm \ref{algo:SPE} works with the help of Example \ref{ex:taxis}. Figure \ref{fig:taxiMACID} shows the original MAID along with its parametrisation and proper MAID subgame. Figure \ref{fig:taxi} shows how we would parameterise a MAIM object for this example to use with our codebase. The MAID subgame ordering will be $\prec~ = \macid_1, \macid_2$ (shown in Figures \ref{fig:taxiMACID} a) and \ref{fig:taxiMACID} c) respectively) because $\macid_2$ is a proper MAID subgame of $\macid_1$. 

The only decision node in $\macid_2$ is $D^2$, and there are two two MAIMs for $\macid_2$ corresponding with the two possible instantiations of $D'1$: $d^1 \in \{e, c\}$. Player 2 will maximise its expected utility in $\macid_2$ by choosing $d^2 = c$ if $d^1 = e$ and $d^2 = e$ if $d^1=c$. In the next iteration of the optimisation loop, we consider the MAIM for $\macid_1$: $D^2$ has already been assigned a decision rule and so is now a chance node with a CPD corresponding to the policy just given. Therefore, $D^1$ is the only decision node remaining in $\macid_1$. Player 1's expected utility in this MAIM subgame is maximised by choosing $e$. Thus, the algorithm will output the only pure policy SPE of the game, shown in Figure \ref{fig:taxiNE} a). 

Our implementation uses a data tree to store results of each successive MAID subgame optimisation. The size of the tree depends on how many decision nodes there are and how many actions are available at each. The leaves of the tree are initially labelled with all the possible combination of actions selected at each of a MAIM’s decision nodes. At each iteration of Algorithm \ref{algo:SPE}, the data tree is reduced once. At each reduction, we find the first node that is empty according to a post-order traversal of the tree and assign it the label of its child which returns the largest expected utility for the agent making that choice. We continue reducing the tree until it is full, then read off the resulting pure policy SPE (Figure \ref{fig:taxi}). 

In this example, we did not have any `ties', however these can arise when comparing the expected utility an  agent  will  get  from  selecting  different  actions.  To account for this, we create a queue system. We pop the tree from the front of our queue, reduce it once and then push it to the back of the queue; each time there is a tie, we push multiple trees (one for each option) to the back of the queue. Our algorithm is complete when all trees in the queue are full. Using this technique, we find all pure policy SPEs. 

\subsection{Gambit} \label{sec:gambit}
An especially useful method belonging to our MACID class converts any MAIM into an EFG that can be used with Gambit \citep{mckelvey2006gambit}. Gambit is a powerful open-source collection of tools for reasoning about games. 
Our method \texttt{MAID\_to\_Gambit\_file} converts a MAID into an EFG using the procedure described in Appendix \ref{sec:macid2efg} before writing this EFG as a \texttt{.efg} file using a prefix-order traversal of the EFG. This file can then run with Gambit's command-line tools as well as with its GUI (Figure \ref{fig:gambit}) to compute different type of Nash Equilibria, to find beliefs at different nodes, or to find each agent's weakly or strongly dominated strategies. 

\end{document}